\documentclass[twocolumn, aps, prl, showpacs, superscriptaddress, longbibliography, nobibnotes]{revtex4-2} % {{{*
\usepackage[utf8]{inputenc}
\usepackage[english]{babel}
\usepackage{amsmath}
\usepackage{amsfonts}
\usepackage{amssymb}
\usepackage{graphicx}
\graphicspath{{img/}{Informe_de_EPS___Mariana/}}
\usepackage{physics}
\usepackage{bbm}
\usepackage[dvipsnames]{xcolor}
\usepackage{booktabs}
\usepackage{float}
\usepackage[normalem]{ulem}
\usepackage{dsfont}

\usepackage{tikz}
\usetikzlibrary{calc,decorations.pathreplacing,positioning}

\usepackage{hyperref}
\hypersetup{
colorlinks=true,
linkcolor=blue,
filecolor=blue,
citecolor=blue,
urlcolor=blue
}
\usepackage[mathlines]{lineno}
\usepackage[draft,inline,nomargin]{fixme} \fxsetup{theme=color}
\definecolor{fxtarget}{HTML}{C8003C}
\definecolor{jacolor}{HTML}{C82800}
\FXRegisterAuthor{ja}{aja}{\color{jacolor}JA}
\FXRegisterAuthor{cp}{acp}{\color{blue}CP}
\definecolor{mpcolor}{RGB}{0,140,140}
\FXRegisterAuthor{mp}{amp}{\color{mpcolor}Mariana}
\FXRegisterAuthor{jn}{ajn}{\color{purple}JN}
\FXRegisterAuthor{cl}{acl}{\color{ForestGreen}CL}
\definecolor{codexcolor}{HTML}{7A1FA2}
\FXRegisterAuthor{codex}{acodex}{\color{codexcolor}Codex}

\NewDocumentCommand{\flexfxnote}{m m m m}{%
  \IfNoValueTF{#3}{#1{#2}}{%
    \IfNoValueTF{#4}{%
      #1{\sout{#2} #3}%
    }{%
      #1{\sout{#2} #3 {\small \tt #4}}%
    }%
  }%
}
\NewCommandCopy{\fxnoteoriginal}{\fxnote}
\NewCommandCopy{\janoteoriginal}{\janote}
\NewCommandCopy{\cpnoteoriginal}{\cpnote}
\NewCommandCopy{\mpnoteoriginal}{\mpnote}
\NewCommandCopy{\jnnoteoriginal}{\jnnote}
\NewCommandCopy{\clnoteoriginal}{\clnote}
\NewCommandCopy{\codexnoteoriginal}{\codexnote}
\RenewDocumentCommand{\fxnote}{m g g}{%
  \flexfxnote{\fxnoteoriginal}{#1}{#2}{#3}}
\RenewDocumentCommand{\janote}{m g g}{%
  \flexfxnote{\janoteoriginal}{#1}{#2}{#3}}
\RenewDocumentCommand{\cpnote}{m g g}{%
  \flexfxnote{\cpnoteoriginal}{#1}{#2}{#3}}
\RenewDocumentCommand{\mpnote}{m g g}{%
  \flexfxnote{\mpnoteoriginal}{#1}{#2}{#3}}
\RenewDocumentCommand{\jnnote}{m g g}{%
  \flexfxnote{\jnnoteoriginal}{#1}{#2}{#3}}
\RenewDocumentCommand{\clnote}{m g g}{%
  \flexfxnote{\clnoteoriginal}{#1}{#2}{#3}}
\RenewDocumentCommand{\codexnote}{m g g}{%
  \flexfxnote{\codexnoteoriginal}{#1}{#2}{#3}}

\newcommand{\mysection}[1]{\par{{\textit{#1---}}}\hspace{-2pt}}

\newcommand{\refsm}[1]{Sec.~\ref{#1} of Supplemental Material~\cite{SM}}

\newcommand{\Eref}[1]{Eq.~(\ref{#1})}

\newcommand{\Fref}[1]{Fig.~\ref{#1}}
\newcommand{\Tref}[1]{Table~\ref{#1}}

\newcommand{\mbZ}{\mathbb{Z}}
\newcommand{\mbC}{\mathbb{C}}
\newcommand{\mcH}{\mathcal{H}}

\newcommand{\mcT}{\mathcal{T}}
\newcommand{\vmcT}{\vec{\mcT}}

\newcommand{\one}{\mathbbm{1}}

\newtheorem{theorem}{Theorem}
\newtheorem{proof}{Proof}

\newcommand{\unam}{Universidad Nacional Aut\'onoma de M\'exico, Ciudad de M\'exico 04510, Mexico}
\newcommand{\ifunam}{Instituto de F\'{\i}sica, \unam}
\newcommand{\atominstitut}{Vienna Center for Quantum Science and Technology, Atominstitut, TU Wien, 1020 Vienna, Austria}
\begin{document}
% Title, authors, abstract {{{*
\newcommand{\titlepaper}{
  The Geometry of Transport in Quantum Walks and Parrondo's Paradox
}
\title{\titlepaper}

\author{Jose Alfredo de Leon}\email{deleongarrido.jose@gmail.com}\affiliation{\ifunam}

\author{Mariana Pérez-Muralles}
\affiliation{Escuela de Ciencias Físicas y Matemáticas, Universidad de San Carlos de Guatemala, Ciudad Universitaria, Guatemala 01012, Guatemala}

\author{Jan Neuser}\affiliation{\atominstitut}

\author{Carlos Pineda}\affiliation{\ifunam}\affiliation{\atominstitut}

\begin{abstract}
We study Parrondo's paradox---the phenomenon where combining losing strategies
yields a winning one---in a minimal discrete-time quantum walk. 
We introduce the transport vector, encoded in the coin's
steady state, whose inner product with the initial coin state gives the
walker's asymptotic velocity.
More generally, the paradox emerges exactly when
the transport vector of the combined strategy falls outside the cone spanned by
the individual ones---a geometric criterion valid for any combination of
strategies. It
explains, for instance, why composing two coin operators within a single step
can produce the paradox while simple alternation between them cannot, since
alternation keeps the combined vector confined to the cone. The paradoxical set
has nonzero measure, and we compute its probability explicitly in
representative cases. This casts the paradox as one instance of designing
reachable transport in quantum walks.
\end{abstract}
%*}}}
\maketitle
\mysection{Introduction} % {{{*
Parrondo's paradox occurs when two processes with the same outcome produce
the opposite outcome when combined~\cite{Parrondo1996,harmer_1999_losing,
wolf_2005_diversity}. 
% If needed, we can delete this sentence:
The paradox has been reported across a wide range of 
disciplines~\cite{abbott_2010_asymmetry,wen_2024_parrondo_review,cheong_2019_paradoxical,cheong_2020_relieving,cheong_2022_lysis,gokhale_2023_crop,spurgin_2005_switching,lai_2020_social_review,mishra_2024_routing,pires_2026_chaos}.
In transport, it manifests when individually biased dynamics are combined
and the resulting motion reverses direction~\cite{amengual_2004_discrete}. Quantum walks provide a natural 
setting for such reversals, since coherent dynamics can interfere under 
composition~\cite{flitney_2012_quantum,jan_2020_experimental}.
This raises two central questions in the discrete-time setting:
can the asymptotic transport of a quantum walk be characterized 
in terms of a simple quantity,
and can such a description determine the general conditions under which
Parrondo's paradox can emerge, as well as its minimal realization? 

Discrete-time quantum walks (DTQWs) provide a simple setting for addressing 
these questions. They describe a lattice walker controlled by an internal
degree of freedom, the \emph{coin}.
In the standard one-dimensional 
model, a two-level coin operator $C$ followed by a conditional shift $S$ defines 
the unitary step $U=SC$~\cite{Aharonov1993,Kempe2003}. 
Different choices of $C$ thus define different walk operators $U_i$, which can
drive the walker to the left or right. Parrondo's paradox occurs
when two walks that individually favor the same direction produce transport in the opposite direction when coherently combined.
The paradox
has been explored in many quantum-walk architectures and
protocols~\cite{flitney2004,kosik_2007_quantum,bulger_2008_positiondependent,chandrashekar_2011_parrondos,flitney_2012_quantum,li_2013_quantum,pawela_2013_cooperative,rajendran_2018_implementing,rajendran_2018_playing,machida_2018_limit,lai_2020_parrondo,lai2020parrondo,jan_2020_experimental,pires_2020_parrondos,walczak_2021_parrondos,walczak_2022_parrondos,trautmann_2022_momentum,jan_2023_territories,mielke_2023_lowdimensional,walczak_2023_noiseinduced,ximenes2024parrondo,kadiri_2024_scouring,mittal_2024_parrondos,walczak2024parrondo,walczak_2025_parrondos,chen_2025_singlequbit,ximenes_2026_aperiodic,cordeiro_2026_defective}. 
% \jnnote{I also believe this should be longer and more explanotory}\janote{For now, I'd like to keep it packaged until we solve the space problem}. 
In fact, Ref.~\cite{flitney_2012_quantum} 
concluded that suitable sequences of two walk operators
can produce such a reversal, although only transiently. 
Persistent reversals were later obtained using larger coin spaces, 
multiple-coin states, or explicit time 
dependence~\cite{rajendran_2018_playing,rajendran_2018_implementing,pires_2020_parrondos}.
Ref.~\cite{jan_2020_experimental} then experimentally realized the paradox in the same 
DTQW setting via period-three step alternation, one of the two 
combination schemes considered here; the reversal was observed over ten 
steps, with simulations supporting its long-time 
persistence. Subsequent numerics found extended 
parameter regions with persistent Parrondo behavior for periods three 
and higher~\cite{jan_2023_territories}. What remained missing was 
a criterion identifying which coherent combinations can reverse the 
asymptotic current and the minimal structure required to do so.
% Parrondo-based quantum-walk protocols have also been explored for entanglement generation, chaos control, quantum search, and cryptography~\cite{panda_2021_order,lai_2021_chaotic,panda2022generating,hosaka_2024_parrondos,rath_2025_chaos,pawela_2025_image,giri_2026_exceptional,rath_2026_crypto}. These results showed that simple DTQWs can sustain Parrondo behavior, but provided no criterion for which coherent combinations can or cannot reverse the asymptotic current, or for the minimal combination.

At the heart of this problem lies a simpler question: what determines the
direction of asymptotic transport? We answer it by introducing the
\emph{transport vector} $\vmcT(U)$, which isolates the 
full step operator $U=SC$
contribution to directed motion.
For the class of one-dimensional two-state DTQWs considered here,
we build on known results for the ballistic 
spreading~\cite{ambainis_2001_onedimensional,konno_2002_quantum,grimmett_2004_weak,konno_2005_limit,ahlbrecht_2011_asymptotic}
and show that the asymptotic velocity of the walker factorizes as
\begin{equation}
\overline{v}=\vmcT(U)\cdot\vec r_0,
\end{equation}
where $\vec r_0$ is the Bloch vector of the initial coin state and 
$\vmcT(U)$ depends only on $U$. 
Thus, $\vmcT(U)$ determines how the initial coin state sets the direction of
asymptotic transport.
We also show that the $z$ component of this stationary coin state, studied in 
Ref.~\cite{annabestani_2020_asymptotic}, is exactly the asymptotic velocity 
$\overline{v}$, giving an alternative, intuitive coin-space picture of the walker's long-time motion.

With the transport vector, Parrondo's paradox becomes a geometric problem of
relating the transport vectors of the individual walks to that of their
coherent combination. We derive a general criterion for when the combined walk
can reverse the common drift of the individual ones, and apply it to two
protocols: \emph{coin composition}, e.g., $U=SC_2C_1$, where coin operators are
combined within a single step, and \emph{step alternation}, e.g., $U_2U_1$,
where complete steps are applied sequentially. The contrast is sharp: coin composition already produces the paradox
with two coin operators and no temporal modulation, whereas two-step alternation can
never do so and a three-step sequence is minimal. Moreover, the effect is not
rare: when the two coin operators and the initial coin state are sampled independently
and uniformly, one in six configurations is paradoxical. Thus, the
transport-vector picture reveals not only when the paradox is possible, but
also its minimal realization and how likely it is to arise, opening a route to
the systematic study of directed transport in quantum walks. 
%*}}}
\mysection{Transport vector in a DTQW} % {{{*
The 1D DTQW describes a \textit{walker} with an internal \textit{coin} degree of freedom, say its spin, moving on a one-dimensional infinite lattice in discrete-time steps~\cite{portugal_2018_quantum, nayak_2000_quantum, ambainis_2003_quantum}.
The system evolves on $\mcH_p \otimes \mcH_c$, where $\mcH_p = \text{span}\{\ket{x}\}_{x \in \mathbb{Z}}$ ($\mathbb{Z}$ denoting the integers) is the position space and $\mcH_c = \mathbb{C}^2$ is the coin space. 
At discrete time $t$, $\ket{\psi_t}=U^t\ket{\psi_0}$, with single-step evolution operator $U=SC$, conditional shift $S = \sum_{x \in \mathbb{Z}} \qty(\ket{x+1}\bra{x} \otimes \dyad{0} + \ket{x-1}\bra{x} \otimes \dyad{1})$, and arbitrary unitary operator $C$ acting on $\mcH_c$.
The $\dyad{0}$ term of $S$ produces a positive lattice displacement, whereas the $\dyad{1}$ term produces a negative one.
In Fourier space, $U = \int \frac{dk}{2\pi} \dyad{k} \otimes \tilde U(k)$, with $\tilde U(k) = e^{-ik\sigma_z} C$~\cite{portugal_2018_quantum}.
For localized initial states, translational invariance lets us set
$\ket{\psi_0} = \ket{x=0} \otimes \ket{c_0}$ without loss of
generality~\cite{portugal_2018_quantum}, where $\ket{c_0}$, the initial state
of the coin, plays a role determining the asymptotic drift.
This is the minimal quantum-walk setting: a two-dimensional coin on a one-dimensional lattice, a fixed coin operator, and an initially separable state.

Unlike a random walker, a quantum walker spreads ballistically.
In Fourier space this follows because $\langle x\rangle(t)$ contains a term linear in time plus subleading terms that do not contribute to the asymptotic velocity; see \refsm{sec:asymptotic_velocity_derivation};
thus, $\langle
x\rangle(t\to\infty)=\overline{v}t$~\cite{portugal_2018_quantum}.
% , where we show that $\overline v$ manifests in both walker and coin subspaces.
The asymptotic velocity factorizes as $\overline{v}=\vmcT\cdot\vec r_0$, where $\vec r_0$ is the Bloch vector of $\ket{c_0}$ and the \textit{transport vector} $\vmcT$ encapsulates the operator's sole contribution to the asymptotic drift; see \refsm{sec:asymptotic_velocity_derivation}. Moreover, we prove that $\overline v$ equals the $z$ component of the coin's steady-state Bloch vector; see \refsm{sec:reduced_coin}.

A key result of this Letter is that the walker's asymptotic drift is encoded in a single initial-state-independent geometric quantity, the transport vector $\vmcT$.
A transport vector can always be defined for translationally invariant DTQWs whose $\tilde U(k)$ has a $k$-independent determinant, including the present setup, alternating coin-operator sequences, and sequences containing higher-order momentum frequencies such as $e^{\pm 2ik}$.
In these cases, the Fourier-space unitary step differs from an $\mathrm{SU}(2)$ matrix only by a $k$-independent global phase, which does not contribute to $\vmcT$; see \refsm{sec:asymptotic_velocity_derivation}. 
Henceforth, $\tilde U(k)$ denotes the corresponding $\mathrm{SU}(2)$ representative of the original Fourier-space unitary step with this phase removed. 
Writing $\tilde U(k)=e^{-i\omega(k)\hat n(k)\cdot\vec\sigma}$, with quasienergy $\omega(k)$ and corresponding Bloch axis $\hat n(k)$, \refsm{sec:asymptotic_velocity_derivation} gives $\vmcT=\int_{-\pi}^{\pi}\frac{\dd k}{2\pi}(\partial_k\omega(k))\hat n(k)$, which can be written directly in terms of $\tilde U(k)$ as 
\begin{equation}\protect\label{eq:transport_vector_Uk}
  \vmcT= i \int_{-\pi}^{\pi}\frac{\dd k}{2\pi}
  \frac{\partial_k\qty{\Tr[\tilde U(k)]}
  \Tr[\vec\sigma\, \tilde U(k)]}
  {\bqty{\Tr[\tilde U(k)]}^2-4}.
\end{equation}
where $\vec{\sigma} = (\sigma_x, \sigma_y, \sigma_z)$ is the vector of
Pauli matrices.
For $\tilde U(k)=e^{-ik\sigma_z}C$, the integral evaluates analytically to the closed-form expression 
\begin{equation}\label{eq:transport_vector_integrated}
  \vmcT = \frac{1}{1+\abs{C_{01}} }
  \begin{pmatrix}
    \Re(C_{01}\, C_{00}^*) \\
    -\Im(C_{01}\, C_{00}^*) \\
    \abs{C_{00}}^2
  \end{pmatrix},
\end{equation}
where $C_{ij}$ denote the matrix elements of $C$;
see alternative derivations in \refsm{sm:transport_vector_integral} 
and \refsm{sm:residue_method}. This reveals three
fundamental properties of the ballistic transport.
First, $\overline v=0$ whenever $C_{00}=0$, proving that coin operators flipping the computational basis yield zero net drift for any initial state.
Second, $\vmcT$ is invariant under the transformation $C\to e^{i\phi}C$;
% this allows us to restrict to $\mathrm{SU}(2)$ coins, .
thus, a U(2) matrix and its SU(2) representative have identical transport vectors.
Third, $\vmcT$ is invariant under $C\to e^{-i\gamma\sigma_z}C$ because $(C_{00},C_{01})\to(C_{00}e^{-i\gamma},C_{01}e^{-i\gamma})$ leaves Eq.~\eqref{eq:transport_vector_integrated} unchanged. Thus, appending a final $z$ rotation to the coin operator does not affect the asymptotic velocity; e.g., $H=ie^{-i\pi/2\,\sigma_z}e^{i\pi/4\,\sigma_y}$ has the same $\vmcT$ as $C=e^{i\pi/4\,\sigma_y}$.

The coin's asymptotic steady state gives a complementary description of the walker's drift.
Tracing out the lattice and taking the long-time limit, the coin relaxes to the stable stationary state
\begin{equation}\label{eq:steady_state}
\rho_C(t\to\infty) = \frac{1}{2} \bqty{ \one + \pqty{ \mathcal{M} \vec{r}_0 } \cdot \vec{\sigma} },
\end{equation}
governed entirely by the symmetric matrix $\mathcal{M} = \int_{-\pi}^{\pi} \frac{\dd{k}}{2\pi} \hat{n}(k) \hat{n}^T(k)$; see \refsm{sec:reduced_coin}.
Since $\partial_k\omega(k)=n_z(k)$, the transport vector is exactly the final column of $\mathcal M$, $\vmcT=\mathcal M\hat z$.
Consequently, the asymptotic velocity is embedded in the coin's relaxation dynamics as the $z$ component of its surviving Bloch vector, $\overline v=r_z(t\to\infty)$. 
Thus, directed spatial transport and internal relaxation are, quite literally, two sides of the same coin.

%*}}}
\mysection{Winning and losing strategies in DTQWs} % {{{*
\begin{figure} % {{{* Figure with the two hemispheres, and the trayectories
  \centering
  \includegraphics[width=\columnwidth]{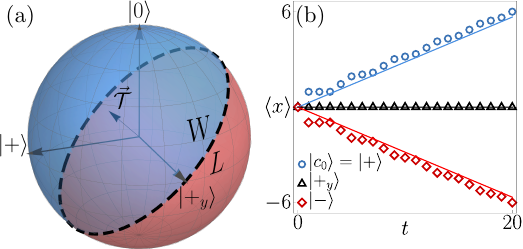}
  \caption{
  (a) Bloch sphere for a Hadamard DTQW, divided by great circle (dashed) 
   into winning (blue, W) and losing (red, L) hemispheres of
  initial states $\vec{r}_0$, based on the sign of the asymptotic velocity
  $\overline{v} = \vmcT \cdot \vec{r}_0$; the arrow marks the transport vector
  $\vmcT$.
  % JA resuelve:
  % \cpnote{quiza dejar los ejes sin la flecha porque no 
  % se distingue del vector de transporte}\janote{thick arrow}. 
  (b) Mean position $\ev{x}$ versus
  time $t$ for three representative initial states, one selected from the
  winning hemisphere, the losing hemisphere, and the null equator in (a):
  $\ket{+}$ (blue), $\ket{+_y}$ (black), and $\ket{-}$ (red), respectively.
  Solid lines denote the asymptotic prediction $\overline{v}t$.
  % JA resuelve:
  % \cpnote{Aca creo que 
  % debemos decirle al lector donde calculamos estas vainas}\janote{footnote}.
  }
  \label{fig:strategies}
\end{figure}%*}}}
To formulate Parrondo's paradox, we define a \textit{strategy} as the choice of a coin operator $C$, specifying the unitary step operator $U$, and an initial coin state $\ket{c_0}$. Analogously to positive capital growth in classical betting, a winning (losing) strategy is a DTQW with positive (negative) asymptotic velocity $\overline v$.

The strategy outcome is equivalently the sign of either $\vmcT\cdot\vec r_0$ or the $z$ component of the coin's steady state:
\begin{subequations}\label{eq:strategies}
\begin{align}
  \text{winning strategy if:} \quad &\vmcT \cdot \vec{r}_0 = r_z(t\to\infty) > 0, \\
  \text{losing strategy if:} \quad &\vmcT \cdot \vec{r}_0 = r_z(t\to\infty) < 0.
\end{align}
\end{subequations}
Crucially, this reconciles two classifications used in the DTQW literature. Finite-time studies define success by either spatial drift~\cite{flitney_2012_quantum,lai_2020_parrondo,pires_2020_parrondos,walczak_2021_parrondos,walczak_2022_parrondos,mittal_2024_parrondos} or an internal coin bias favoring $\ket0$ over $\ket1$~\cite{chandrashekar_2011_parrondos,rajendran_2018_implementing,rajendran_2018_playing,jan_2020_experimental,jan_2023_territories}. Equations~\eqref{eq:strategies} establish that, in the long-time limit, the signs of the walker's mean position and the coin's internal state probabilities are strictly equivalent.

Geometrically, once $U$ and thus $\vmcT$ are fixed,
Eq.~\eqref{eq:strategies} makes the outcome depend only on the orientation
of $\vec r_0$ relative to $\vmcT$. The plane
$\vmcT\cdot\vec r_0=0$ defines an equator of strategies, bisecting the
initial-state Bloch sphere into winning (``northern'') and losing
(``southern'') hemispheres, shown by the dichromatic coloring for a
Hadamard DTQW in Fig.~\ref{fig:strategies}(a).
Figure~\ref{fig:strategies}(b) shows the asymptotic
$\ev{x}(t)$ for states from each domain. A strategy is null when
$\vec r_0\perp\vmcT$ or $\vmcT=\vec0$.

This geometric classification underlies our analysis of Parrondo's paradox in DTQWs, 
defined here as reversing the walker's asymptotic drift by combining 
strategies with the same individual outcome.
For evolution operators $U_{1,2}=SC_{1,2}$ and initial state $\vec r_0$, each 
strategy is losing (winning) when $\vmcT(U_i)\cdot\vec r_0<0$ 
[$\vmcT(U_i)\cdot\vec r_0>0$]. The paradox arises when this holds for 
both $U_i$, while another evolution operator constructed from them---e.g., 
$SC_2C_1$ or $U_2U_2U_1$---has the opposite transport-vector dot-product 
sign against the same $\vec r_0$.

To investigate these inequalities, we analyze two simple, distinct ways of 
combining strategies commonly used in DTQWs.
In \emph{coin composition}, coin operators are combined before the 
shift~\cite{chandrashekar_2011_parrondos,mittal_2024_parrondos,lai_2020_parrondo}; 
within this case we also single out rotation compositions about a common axis. 
In \emph{step alternation}, complete single-step unitary operators are applied sequentially~\cite{chandrashekar_2011_parrondos,flitney_2012_quantum,rajendran_2018_implementing,rajendran_2018_playing,pires_2020_parrondos,jan_2020_experimental,walczak_2021_parrondos,walczak_2022_parrondos,jan_2023_territories}.
Both composite operators' transport vectors follow directly from Eq.~\eqref{eq:transport_vector_Uk}, whereas Eq.~\eqref{eq:transport_vector_integrated} applies only to coin composition. The paradox is then determined by evaluating these composite vectors against the strategy inequalities.

%*}}}
\mysection{Coin composition} % {{{*
In coin composition, coin operators are applied consecutively before 
each shift, the simplest case being $U=SC_2C_1$. 
This two-coin-operator composition already produces Parrondo's paradox.
Consider $C_i(\hat n,\chi_i)=\exp[-i(\chi_i/2)\hat n\cdot\vec\sigma]$, 
rotations by $\chi_i$ about a common axis $\hat n$, and an equatorial 
initial coin state with Bloch vector $\vec r_0$.
For fixed $\vec r_0$ and $\hat n$, a strategy changes outcome where its 
asymptotic velocity vanishes, i.e., where the transport vector is orthogonal 
to the initial state, $\vmcT(\chi)\cdot\vec r_0=0$. 
Besides the trivial solution $\chi=0$, generic configurations have a 
unique critical angle $\chi_{\mathrm c}\in(0,2\pi)$; 
see \refsm{SM:critical_rot_angle}. 
These two zeros divide $[0,2\pi)$ into winning and losing intervals with 
positive and negative $\overline v$, respectively, as shown in 
\Fref{fig:PhaseDiagrams}.
Since common-axis rotations are additive, $C_2C_1=C(\hat n,\chi_1+\chi_2)$. 
Thus, Parrondo's paradox occurs when $\chi_1$ and $\chi_2$ lie in the 
same interval but $(\chi_1+\chi_2)\bmod 2\pi$ lies in the opposite one, 
moving from a winning interval to a losing one or vice versa.

For an equatorial initial coin state and two coin operators sharing a 
common axis, the inset of \Fref{fig:PhaseDiagrams} gives a 
phase-diagram-like outcome map in the $(\chi_1,\chi_2)$ plane.
Each individual outcome is determined by which side of $\chi_{\mathrm c}$ 
contains $\chi_i$, while the composition is classified by 
$(\chi_1+\chi_2)\bmod2\pi$. These three binary classifications partition 
the plane into regions such as $L\circ L=L$ and $L\circ L=W$.
The inset highlights regions where both individual outcomes are the same, 
distinguishing intuitive $L\circ L=L$ and $W\circ W=W$ from stippled 
paradoxical $L\circ L=W$ and $W\circ W=L$ regions.

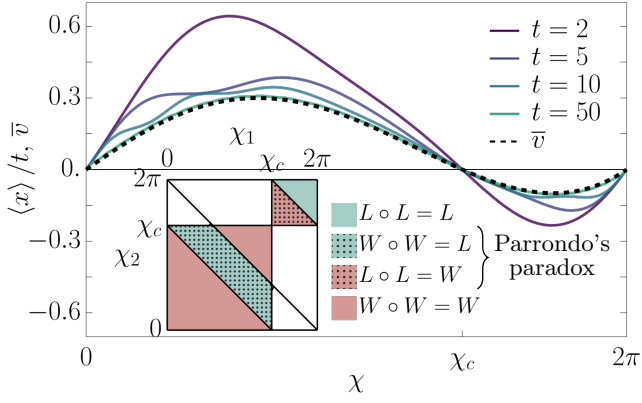
\begin{figure} % {{{* Asymptitc velocity vs \chi_c and phase diagram
  \centering
  % TikZ's coordinate calculations use restricted horizontal mode, which
  % conflicts with the line-numbering hooks inserted inside the picture.
  % \nolinenumbers
  \begin{tikzpicture}
    \node[inner sep=0] (img) {%
      \includegraphics[width=\columnwidth]{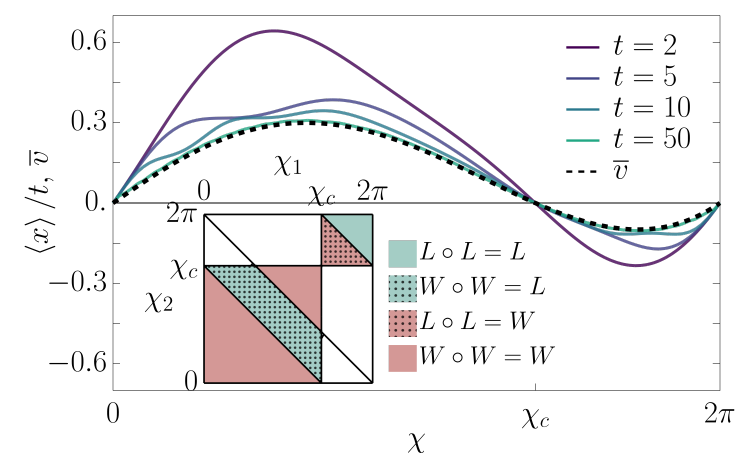}%
    };

    \coordinate (Rbot) at ($(img.south west)!0.75!(img.south east)$);
    \coordinate (Rtop) at ($(img.north west)!0.75!(img.north east)$);

    \coordinate (bracebot) at ($(Rbot)!0.28!(Rtop)$);
    \coordinate (bracetop) at ($(Rbot)!0.43!(Rtop)$);

    \draw[
      decorate,
      decoration={brace,mirror,amplitude=4pt},
      line width=0.5pt
    ]
      (bracebot) -- (bracetop)
      node[
        midway,
        right=3pt,
        font=\fontsize{9pt}{9pt}\selectfont,
        align=center
      ] {Parrondo's\\paradox};
  \end{tikzpicture}
  % \linenumbers
  \caption{DTQWs with initial coin state
  $\ket{c_0} = (\ket{0} + e^{i\pi/4}\ket{1})/\sqrt{2}$ and coin
  operators $C(\hat n,\chi)$ corresponding to Bloch-sphere rotations
  about the Hadamard axis $\hat n = (1,0,1)/\sqrt{2}$.
  Main panel: Mean velocity $\langle x \rangle / t$ as a function of the
  rotation angle $\chi$ for different times $t$ (solid colored
  lines), alongside the analytical asymptotic velocity
  $\overline v(\chi)$ (dashed black line), derived in \refsm{SM:expr_fig3}. The critical angle
  $\chi_\mathrm{c} = 2\left[\pi - \arctan\left(\sqrt{2}\right)\right]$
  separates winning and losing intervals. Inset: outcome map for two coin
  operators $C(\hat n,\chi_i)$ sharing the same axis $\hat n$.
  Highlighted are the intuitive regions
  ($L \circ L = L$ and $W \circ W = W$) and, stippled, the paradoxical
  regions ($L \circ L = W$ and $W \circ W = L$).}
  \label{fig:PhaseDiagrams}
\end{figure}% *}}}

Parrondo's paradox requires no fine tuning in this simple composition 
setup.
For fixed $\vec r_0$ and $\hat n$, with $\chi_1$ and $\chi_2$ drawn 
independently and uniformly from $[0,2\pi)$, a configuration outcome 
such as $L\circ W=W$ has probability equal to its area in the $(\chi_1,\chi_2)$ 
outcome map divided by $4\pi^2$.
Summing the paradoxical regions gives 
$\mathbb{P}(L\circ L=W\mid\vec r_0,\hat n) + 
\mathbb{P}(W\circ W=L\mid\vec r_0,\hat n) =
\chi_{\mathrm c}(2\pi-\chi_{\mathrm c})/4\pi^2$, whereas summing 
the intuitive regions  gives 
$\mathbb{P}(L\circ L=L\mid\vec r_0,\hat n) + 
\mathbb{P}(W\circ W=W\mid\vec r_0,\hat n) =
(4\pi^2-6\pi\chi_{\mathrm c}+3\chi_{\mathrm c}^2)/4\pi^2$. For the 
parameters of \Fref{fig:PhaseDiagrams}, these probabilities are 
approximately $0.212$ and $0.365$, respectively.  
More generally, the paradoxical-to-intuitive ratio remains finite and 
nonzero for $0<\chi_{\mathrm c}<2\pi$, equals unity at $\chi_{\mathrm c}=\pi$, 
and vanishes only as $\chi_{\mathrm c}\to0^+$ or $\chi_{\mathrm c}\to(2\pi)^-$, 
where the paradoxical regions disappear. 
Does Parrondo's paradox remain prevalent without the common-axis and 
equatorial-initial-state restrictions?

It does. To determine how often it occurs without these restrictions, 
we draw $C_1$, $C_2$, and $\vec r_0$ independently according to their 
invariant Haar measures.
The two paradoxical configurations have total probability $1/6$; 
see \refsm{sec:joint_prob_derivation}. The corresponding intuitive
configurations have probability $1/2$, only three times as often.
Overall, nonparadoxical configurations have probability $5/6$ and 
are only five times as likely as paradoxical ones. This is remarkably 
large: even in this minimal setting, the paradox arises naturally and 
frequently rather than requiring carefully engineered DTQWs.

%*}}}
\mysection{Step alternation} % {{{*
Step alternation instead applies complete steps $U_i=SC_i$ sequentially, 
rather than combining coin operators within a single step, as in the 
experimentally realized $ABB$ sequence~\cite{jan_2020_experimental}.
The minimal instance, $U_2U_1$, cannot yield the paradox.
As shown analytically in \refsm{sm:colinearity}, 
$\vmcT(U_2U_1)=\gamma\vmcT(U_1)$ with $\gamma\geq0$. 
Hence, if $U_1$ is losing for $\vec r_0$, $\vmcT(U_1)\cdot\vec r_0<0$, 
then $\vmcT(U_2U_1)\cdot\vec r_0\leq0$; the alternation is losing for 
$\gamma>0$ and null for $\gamma=0$. Thus, period two cannot reverse 
the asymptotic drift for any initial coin state.

The next candidate is therefore a three-step alternation; 
we consider $U_2U_2U_1$, matching the $ABB$ pattern of the experimental
architecture realized.
With a third factor, $\vmcT(U_2U_2U_1)$ is no longer forced to be 
collinear with $\vmcT(U_1)$, as made precise in the next section 
for an arbitrary number of strategies. 
Figure~\ref{fig:parrondo_emergence}(a) shows such a case: the red stippled 
region contains initial coin states for which $U_1$ and $U_2$ are 
losing while $U_2U_2U_1$ is winning. Panel (b) confirms the 
corresponding asymptotic drift reversal for one such state.

\begin{figure} % {{{* Parrondo areas in Bloch sphere for step alternation
  \centering
  \includegraphics[width=\columnwidth]{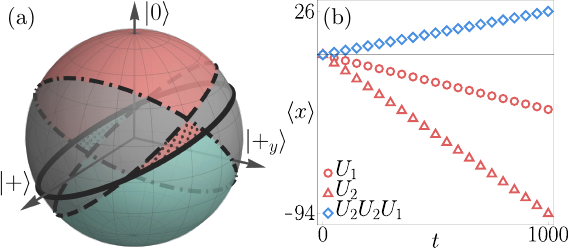}
  % Old caption {{{*
  % \caption{
  % (a) Winning/losing boundaries on the Bloch sphere of initial states $\vec{r}_0$ for the individual strategies $U_1$ and $U_2$ and for the three-step alternation $U_2U_2U_1$, shown as three great circles. The equators, defined by $\vmcT(U_i)\cdot \vec r_0 = 0$, are shown as dashed, dash dotted, and solid black lines for $U_1$, $U_2$ and $U_2U_2U_1$, respectively.
  % \cpnote{Preguntar a Jose Alfredo qué estilo de línea corresponde a cada operador (son arbitrarios en el script de la figura), para citarlo aquí explícitamente}\janote{Lo checo}. 
  % The distinct regions are painted following the same of the inset of \Fref{fig:PhaseDiagrams}, in particular, $L\circ L=W$ is the stippled red region. (b) $\langle x\rangle(t)$ for an initial state selected from the $L\circ L=W$ region (stipped red), showing negative asymptotic drift for $U_1$ and $U_2$ and positive drift for $U_2U_2U_1$.
  % \cpnote{Un poco redundante, ver si deamos el U2U2U1 o el rollo completo con las SCSCSC. Discutir con JA}\janote{cambiar legend para U1U2U2}\mpnote{\tt También pienso que sería bueno cambiar los colores de los círculos mayores del panel de la izquierda por los colores correspondientes en el panel de la derecha. :D}\janote{editar y enviar junto con step-alternation}
  % } *}}}
  \caption{(a) Outcome map on the Bloch sphere of initial states $\vec r_0$
  for $U_1$, $U_2$ randomly chosen, and their three-step alternation $U_2U_2U_1$. The zero-drift 
  great circles $\overline v =\vmcT(U)\cdot\vec r_0=0$ partition the Bloch 
  sphere into winning and losing initial states for $U_1$ (dashed), $U_2$ 
  (dash-dotted), and $U_2U_2U_1$ (solid). Solid and stippled regions indicate 
  intuitive and paradoxical outcomes, respectively, following the colouring 
  of the inset in \Fref{fig:PhaseDiagrams}. In particular, the stippled red 
  region corresponds to $L\circ L=W$.
  (b) Mean position for an initial state in this region: $U_1$ and $U_2$ are 
  losing (open red markers), whereas the minimal three-step alternation 
  $U_2U_2U_1$ is winning (open blue diamonds).
  % \janote{Moví bastante de esta caption. Carlos revisa}
  }
  \label{fig:parrondo_emergence}
\end{figure} %*}}}

We next ask how likely Parrondo's paradox is to emerge for a fixed pair
$C_1,C_2$, sampling the initial coin state uniformly over the Bloch sphere.
% The event $L\circ L=W$ occurs for the states $\vec r_0$ satisfying
% $\vmcT(U_1)\cdot\vec r_0<0$, $\vmcT(U_2)\cdot\vec r_0<0$, and
% $\vmcT(U_2U_2U_1)\cdot\vec r_0>0$ --- the green region of
% Fig.~\ref{fig:parrondo_emergence}(a). Therefore, $\mathbb{P}(L\circ L=W\mid
% U_1,U_2)$ is given by the finite fraction of the Bloch sphere occupied by it,
% namely its area divided by $4\pi$.
For fixed $C_1,C_2$, sampling the initial coin state uniformly over the Bloch sphere, $L\circ L=W$ occurs when $\vmcT(U_1)\cdot\vec r_0<0$, $\vmcT(U_2)\cdot\vec r_0<0$, and $\vmcT(U_2U_2U_1)\cdot\vec r_0>0$, defining the red stippled region of Fig.~\ref{fig:parrondo_emergence}(a). Thus, $\mathbb{P}(L\circ L=W\mid U_1,U_2)$ is the fraction of the Bloch sphere occupied by this region, its area divided by $4\pi$.

A complementary probability conditions on the two individual 
strategies already being losing: restricting the sampling to the 
intersection of the losing hemispheres of $U_1$ and $U_2$, it is 
the red stippled area divided by the area of that $L\circ L$ intersection. 
For fixed $C_1$ and $C_2$, the paradoxical states still occupy a finite 
Bloch-sphere region, so a suitable pair of coin operators requires no 
fine tuning of the initial coin state. These areas can, in principle, 
be obtained from the residue-theorem method of \refsm{sm:residue_method}; 
for alternations such as $U_2U_2U_1$, however, this requires finding 
the roots of a degree-six polynomial and determining whether they lie 
inside the unit circle, beyond the scope of our analytical analysis.

%*}}}
\mysection{General criterion for Parrondo's paradox} % {{{*
\label{sec:general_criterion}
The transport-vector picture unifies the preceding cases and extends 
directly to $N$ strategies.
Let $U_i=SC_i$ have transport vectors $\vmcT_i=\vmcT(U_i)$, and let 
$U_{\mathcal W}$ denote any combination built from them through coin 
composition or step alternation. Let $L^{\circ N}$ ($W^{\circ N}$) 
denote $N$ individual losing (winning) strategies. For $L^{\circ N}=W$ 
or $W^{\circ N}=L$, $\vmcT(U_{\mathcal W})\cdot\vec r_0$ must have 
the opposite sign from $\vmcT_i\cdot\vec r_0$ for every $i$.
Let $\mathcal K=\operatorname{cone}\{\vmcT_1,\ldots,\vmcT_N\}
=\left\{\sum_{i=1}^N a_i\vmcT_i
\,\middle|\,a_i\in\mathbb R_{\geq0}\right\}$
be the convex cone generated by the individual transport vectors.
If no coin operator $C_i$ is antidiagonal in the computational basis, \Eref{eq:transport_vector_integrated} gives $(\vmcT_i)_z=\abs{(C_i)_{00}}^2/[1+\abs{(C_i)_{01}}]>0$ for every $i$. Hence, $\mathcal K$ is pointed, i.e., it cannot contain a nonzero vector and its opposite; see Fig.~\ref{fig:generalized_PP}(a).
If $\vmcT(U_{\mathcal W})\in\mathcal K$, then
$\vmcT(U_{\mathcal W})=\sum_{i=1}^N a_i\vmcT_i$ with $a_i\geq0$.
Therefore, whenever all $\vmcT_i\cdot\vec r_0$ have the same sign, their
nonnegative linear combination cannot have the opposite sign, and the paradox
is impossible.
If instead $\vmcT(U_{\mathcal W})\notin\mathcal K$, the pointedness of $\mathcal K$
guarantees the existence of a plane through the origin that strictly separates
$\vmcT(U_{\mathcal W})$ from the cone~\footnote{As a finitely generated cone,
$\mathcal K$ is closed and convex. Since
$\vmcT(U_{\mathcal W})\notin\mathcal K$, the separating-hyperplane
theorem~\cite{rockafellar_1970_convex} provides a vector $\vec m$ such that
$\vec m\cdot\vmcT(U_{\mathcal W})>0$ and
$\vec m\cdot\vec v\leq0$ for every $\vec v\in\mathcal K$. Because
$(\vmcT_i)_z>0$ for every generator, setting
$\vec m_\epsilon=\vec m-\epsilon\hat z$ for sufficiently small
$\epsilon>0$ gives $\vec m_\epsilon\cdot\vmcT_i<0$ for every $i$ while
preserving $\vec m_\epsilon\cdot\vmcT(U_{\mathcal W})>0$. Normalizing
$\vec m_\epsilon$ gives the unit normal $\hat n$.},
% shown from the edge as the black line in \Fref{fig:generalized_PP}(b), 
with $\mathcal K$ and $\vmcT(U_{\mathcal W})$ in opposite hemispheres.
Consider its normal $\hat n$ pointing toward the same hemisphere as 
$\vmcT(U_{\mathcal W})$, so choosing $\vec r_0=\hat n$ gives 
$\vmcT(U_{\mathcal W})\cdot\hat n>0$  and $\vmcT_i\cdot\hat n<0$ 
for every $i$. A full picture of this construction is presented in  \Fref{fig:generalized_PP}(b). Hence, $\hat n$ and $-\hat n$
define initial states 
realizing $L^{\circ N}=W$ and $W^{\circ N}=L$, respectively, yielding
\begin{equation}
  \text{Parrondo's paradox emerges}
  \,\Leftrightarrow\,
  \vmcT(U_{\mathcal W})\notin\mathcal K.
\end{equation}
Moreover, a paradoxical configuration is robust: continuity of the transport vectors preserves the strict separating inequalities under sufficiently small parameter variations. Therefore, any paradoxical set contains an open neighborhood and consequently has nonzero measure.

\begin{figure} %{{{* Convex cone
  \centering
  \includegraphics[width=\columnwidth]{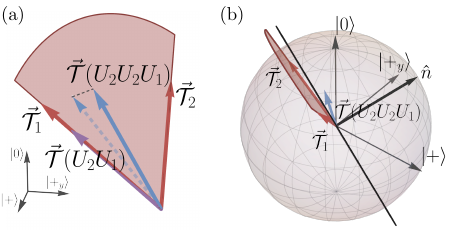}
  \caption{(a) Pointed convex cone $\mathcal K=\operatorname{cone}\{\vmcT_1,\vmcT_2\}$ 
  generated by the individual transport vectors $\vmcT_i=\vmcT(U_i)$, with
  $U_1$ and $U_2$ the same operators as in \Fref{fig:parrondo_emergence}. 
  The transport vector $\vmcT_4=\vmcT(U_2U_1)$ lies inside $\mathcal K$, 
  collinear with $\vmcT_1$, whereas 
  $\vmcT_3=\vmcT(U_2U_2U_1)$ lies outside. 
  (b) The same cone embedded in the Bloch sphere. A plane through the origin 
  separates $\vmcT_3$ from $\mathcal K$; it is viewed edge-on and therefore 
  appears as a black line. Its normal $\hat n$ defines an initial-state 
  direction for which $\vmcT_3\cdot\hat n$ has the opposite sign to 
  $\vmcT_i\cdot\hat n$ for every $i$.
  }
  \label{fig:generalized_PP}
\end{figure} % *}}}

This criterion shows that the protocols differ geometrically in how 
they constrain the combined transport vector. For period-two step 
alternation, $\vmcT(U_2U_1)=\gamma\vmcT(U_1)$ with $\gamma\geq0$, 
so the combined transport vector necessarily remains contained in 
$\mathcal K$, as in \Fref{fig:generalized_PP}(a), and Parrondo's 
paradox is impossible.
A third step removes this collinearity constraint: $\vmcT(U_2U_2U_1)$ 
can leave $\mathcal K$ and be separated from the cone as in 
Fig.~\ref{fig:generalized_PP}(b), opening the paradoxical region in 
Fig.~\ref{fig:parrondo_emergence}.
Coin composition has no analogous period-two obstruction, since 
already $\vmcT(SC_2C_1)$ can lie outside the cone spanned by 
$\vmcT(U_1)$ and $\vmcT(U_2)$.
Thus, the protocols impose different geometric constraints on the 
resulting transport vector while sharing the same criterion for 
Parrondo's paradox.

%*}}}
\mysection{Concluding remarks} % {{{*
By combining dynamics that individually drive the system in one direction, 
we show in this Letter how directed transport can be reversed in the simplest 
discrete-time quantum walk---a walker on a line steered by a two-level 
coin. This reversal is a transport instance of
Parrondo's paradox---a counterintuitive effect, encountered in settings
ranging from game theory to
ratchets~\cite{abbott_2010_asymmetry,amengual_2004_discrete}, in which a
combination overturns the outcomes of its constituents. We demonstrate 
the walker's long-time drift depends only on (i) the initial state of the coin,
and (ii) the transport vector, a geometric object that captures the contribution
of the evolution operator alone. Remarkably, this vector is also imprinted in
the stationary state of the coin, linking the walker's eventual motion to its
internal dynamics.
Expressed in terms of transport vectors, the emergence of Parrondo's
paradox follows a simple geometric criterion: reversal is possible if and 
only if the composite transport vector escapes the convex cone generated by the
individual ones.
This criterion establishes that the experimentally observed
reversals~\cite{jan_2020_experimental} persist asymptotically without
requiring a setup more complex than the standard one-dimensional,
two-state walk, and identifies the minimal realization of each protocol:
two coins for coin composition and a three-step period for step alternation.
Minimality, however, does not mean rarity: in one of the simplest settings
we consider, directly composing two coin operators,
one in every six uniformly sampled configurations produces a
paradoxical reversal, and for more complex compositions it remains robust
under parameter variations. 

More broadly, the transport vector can be identified with the Bloch
representation of the Brillouin-zone average of the group-velocity operator,
whose formulation applies to translation-invariant quantum
walks with bounded jump distances in arbitrary lattice dimension and with finite-dimensional
coins~\cite{ahlbrecht_2011_asymptotic}. This perspective suggests an
inverse-design problem: given a set of available walk operations, which
asymptotic currents are reachable by coherent combination, and what is the
shortest sequence that realizes a target current? Such a program would extend
the control-theoretic study of reachable quantum-walk
states~\cite{albertini_2012_controllability} to reachable transport responses.
In higher-dimensional lattices, the analogous question would concern
directional steering. The Parrondo effect studied here is therefore a minimal
instance of a broader design principle: coherent combination can realize a
transport direction that is inaccessible under every constituent evolution
for the same initial state. 
%*}}}
\mysection{Acknowledgments} % {{{*
Fruitful discussions with Amado Cabrera are acknowledged. 
J.A.d.L., C.P., and M.P.-M. thank Rodolfo Samayoa for his administrative support as local advisor of M.P.-M.
J.A.d.L. acknowledges SECIHTI through funding for graduate studies.
% C.P. acknowledges support by UNAM-PAPIIT IG101324 and SECIHTI CBF-2025-I-1548.
J.N. acknowledges support from the European Research Council (ERC Project
`Cocoquest’ 101043705).
C.P. acknowledges support by UNAM-PAPIIT IG101324, SECIHTI CBF-2025-I-1548 and UNAM PASPA–DGAPA.
C.P. acknowledges financial support from the Austrian Federal Ministry of Education, Science and Research via the Austrian Research Promotion Agency (FFG) through the project FO999921415 (Vanessa-QC) funded by the European Union{\textemdash}NextGenerationEU.
% Fruitful discussions with Amado Cabrera and Vikash Mittal are acknowledged.
%*}}}
\bibliographystyle{apsrev4-2}\bibliography{references}

\begin{thebibliography}{59}%
\makeatletter
\providecommand \@ifxundefined [1]{%
 \@ifx{#1\undefined}
}%
\providecommand \@ifnum [1]{%
 \ifnum #1\expandafter \@firstoftwo
 \else \expandafter \@secondoftwo
 \fi
}%
\providecommand \@ifx [1]{%
 \ifx #1\expandafter \@firstoftwo
 \else \expandafter \@secondoftwo
 \fi
}%
\providecommand \natexlab [1]{#1}%
\providecommand \enquote  [1]{``#1''}%
\providecommand \bibnamefont  [1]{#1}%
\providecommand \bibfnamefont [1]{#1}%
\providecommand \citenamefont [1]{#1}%
\providecommand \href@noop [0]{\@secondoftwo}%
\providecommand \href [0]{\begingroup \@sanitize@url \@href}%
\providecommand \@href[1]{\@@startlink{#1}\@@href}%
\providecommand \@@href[1]{\endgroup#1\@@endlink}%
\providecommand \@sanitize@url [0]{\catcode `\\12\catcode `\$12\catcode `\&12\catcode `\#12\catcode `\^12\catcode `\_12\catcode `\%12\relax}%
\providecommand \@@startlink[1]{}%
\providecommand \@@endlink[0]{}%
\providecommand \url  [0]{\begingroup\@sanitize@url \@url }%
\providecommand \@url [1]{\endgroup\@href {#1}{\urlprefix }}%
\providecommand \urlprefix  [0]{URL }%
\providecommand \Eprint [0]{\href }%
\providecommand \doibase [0]{https://doi.org/}%
\providecommand \selectlanguage [0]{\@gobble}%
\providecommand \bibinfo  [0]{\@secondoftwo}%
\providecommand \bibfield  [0]{\@secondoftwo}%
\providecommand \translation [1]{[#1]}%
\providecommand \BibitemOpen [0]{}%
\providecommand \bibitemStop [0]{}%
\providecommand \bibitemNoStop [0]{.\EOS\space}%
\providecommand \EOS [0]{\spacefactor3000\relax}%
\providecommand \BibitemShut  [1]{\csname bibitem#1\endcsname}%
\let\auto@bib@innerbib\@empty
%</preamble>
\bibitem [{\citenamefont {Parrondo}(1996)}]{Parrondo1996}%
  \BibitemOpen
  \bibfield  {author} {\bibinfo {author} {\bibfnamefont {J.~M.~R.}\ \bibnamefont {Parrondo}}} (\bibinfo {year} {1996}),\ \bibinfo {note} {presentation at the EEC HC\&M Network on Complexity and Chaos (\#ERBCHRX-CT940546), ISI, Torino, Italy. Unpublished}\BibitemShut {NoStop}%
\bibitem [{\citenamefont {Harmer}\ and\ \citenamefont {Abbott}(1999)}]{harmer_1999_losing}%
  \BibitemOpen
  \bibfield  {author} {\bibinfo {author} {\bibfnamefont {G.~P.}\ \bibnamefont {Harmer}}\ and\ \bibinfo {author} {\bibfnamefont {D.}~\bibnamefont {Abbott}},\ }\href {https://doi.org/10.1038/47220} {\bibfield  {journal} {\bibinfo  {journal} {Nature}\ }\textbf {\bibinfo {volume} {402}},\ \bibinfo {pages} {864} (\bibinfo {year} {1999})}\BibitemShut {NoStop}%
\bibitem [{\citenamefont {Wolf}\ \emph {et~al.}(2005)\citenamefont {Wolf}, \citenamefont {Vazirani},\ and\ \citenamefont {Arkin}}]{wolf_2005_diversity}%
  \BibitemOpen
  \bibfield  {author} {\bibinfo {author} {\bibfnamefont {D.~M.}\ \bibnamefont {Wolf}}, \bibinfo {author} {\bibfnamefont {V.~V.}\ \bibnamefont {Vazirani}},\ and\ \bibinfo {author} {\bibfnamefont {A.~P.}\ \bibnamefont {Arkin}},\ }\href {https://doi.org/10.1016/j.jtbi.2004.11.020} {\bibfield  {journal} {\bibinfo  {journal} {Journal of Theoretical Biology}\ }\textbf {\bibinfo {volume} {234}},\ \bibinfo {pages} {227} (\bibinfo {year} {2005})}\BibitemShut {NoStop}%
\bibitem [{\citenamefont {Abbott}(2010)}]{abbott_2010_asymmetry}%
  \BibitemOpen
  \bibfield  {author} {\bibinfo {author} {\bibfnamefont {D.}~\bibnamefont {Abbott}},\ }\href {https://doi.org/10.1142/S0219477510000010} {\bibfield  {journal} {\bibinfo  {journal} {Fluctuation and Noise Letters}\ }\textbf {\bibinfo {volume} {09}},\ \bibinfo {pages} {129} (\bibinfo {year} {2010})}\BibitemShut {NoStop}%
\bibitem [{\citenamefont {Wen}\ and\ \citenamefont {Cheong}(2024)}]{wen_2024_parrondo_review}%
  \BibitemOpen
  \bibfield  {author} {\bibinfo {author} {\bibfnamefont {T.}~\bibnamefont {Wen}}\ and\ \bibinfo {author} {\bibfnamefont {K.~H.}\ \bibnamefont {Cheong}},\ }\href {https://doi.org/10.1016/j.plrev.2024.08.002} {\bibfield  {journal} {\bibinfo  {journal} {Physics of Life Reviews}\ }\textbf {\bibinfo {volume} {51}},\ \bibinfo {pages} {33} (\bibinfo {year} {2024})}\BibitemShut {NoStop}%
\bibitem [{\citenamefont {Cheong}\ \emph {et~al.}(2019)\citenamefont {Cheong}, \citenamefont {Koh},\ and\ \citenamefont {Jones}}]{cheong_2019_paradoxical}%
  \BibitemOpen
  \bibfield  {author} {\bibinfo {author} {\bibfnamefont {K.~H.}\ \bibnamefont {Cheong}}, \bibinfo {author} {\bibfnamefont {J.~M.}\ \bibnamefont {Koh}},\ and\ \bibinfo {author} {\bibfnamefont {M.~C.}\ \bibnamefont {Jones}},\ }\href {https://doi.org/10.1002/bies.201900027} {\bibfield  {journal} {\bibinfo  {journal} {BioEssays}\ }\textbf {\bibinfo {volume} {41}},\ \bibinfo {pages} {1900027} (\bibinfo {year} {2019})}\BibitemShut {NoStop}%
\bibitem [{\citenamefont {Cheong}\ \emph {et~al.}(2020)\citenamefont {Cheong}, \citenamefont {Wen},\ and\ \citenamefont {Lai}}]{cheong_2020_relieving}%
  \BibitemOpen
  \bibfield  {author} {\bibinfo {author} {\bibfnamefont {K.~H.}\ \bibnamefont {Cheong}}, \bibinfo {author} {\bibfnamefont {T.}~\bibnamefont {Wen}},\ and\ \bibinfo {author} {\bibfnamefont {J.~W.}\ \bibnamefont {Lai}},\ }\href {https://doi.org/10.1002/advs.202002324} {\bibfield  {journal} {\bibinfo  {journal} {Advanced Science}\ }\textbf {\bibinfo {volume} {7}},\ \bibinfo {pages} {2002324} (\bibinfo {year} {2020})}\BibitemShut {NoStop}%
\bibitem [{\citenamefont {Cheong}\ \emph {et~al.}(2022)\citenamefont {Cheong}, \citenamefont {Wen}, \citenamefont {Benler}, \citenamefont {Koh},\ and\ \citenamefont {Koonin}}]{cheong_2022_lysis}%
  \BibitemOpen
  \bibfield  {author} {\bibinfo {author} {\bibfnamefont {K.~H.}\ \bibnamefont {Cheong}}, \bibinfo {author} {\bibfnamefont {T.}~\bibnamefont {Wen}}, \bibinfo {author} {\bibfnamefont {S.}~\bibnamefont {Benler}}, \bibinfo {author} {\bibfnamefont {J.~M.}\ \bibnamefont {Koh}},\ and\ \bibinfo {author} {\bibfnamefont {E.~V.}\ \bibnamefont {Koonin}},\ }\href {https://doi.org/10.1073/pnas.2115145119} {\bibfield  {journal} {\bibinfo  {journal} {Proceedings of the National Academy of Sciences}\ }\textbf {\bibinfo {volume} {119}},\ \bibinfo {pages} {e2115145119} (\bibinfo {year} {2022})}\BibitemShut {NoStop}%
\bibitem [{\citenamefont {Gokhale}\ and\ \citenamefont {Sharma}(2023)}]{gokhale_2023_crop}%
  \BibitemOpen
  \bibfield  {author} {\bibinfo {author} {\bibfnamefont {C.~S.}\ \bibnamefont {Gokhale}}\ and\ \bibinfo {author} {\bibfnamefont {N.}~\bibnamefont {Sharma}},\ }\href {https://doi.org/10.1098/rsos.221401} {\bibfield  {journal} {\bibinfo  {journal} {Royal Society Open Science}\ }\textbf {\bibinfo {volume} {10}},\ \bibinfo {pages} {221401} (\bibinfo {year} {2023})}\BibitemShut {NoStop}%
\bibitem [{\citenamefont {Spurgin}\ and\ \citenamefont {Tamarkin}(2005)}]{spurgin_2005_switching}%
  \BibitemOpen
  \bibfield  {author} {\bibinfo {author} {\bibfnamefont {R.}~\bibnamefont {Spurgin}}\ and\ \bibinfo {author} {\bibfnamefont {M.}~\bibnamefont {Tamarkin}},\ }\href {https://doi.org/10.1207/s15427579jpfm0601_3} {\bibfield  {journal} {\bibinfo  {journal} {Journal of Behavioral Finance}\ }\textbf {\bibinfo {volume} {6}},\ \bibinfo {pages} {15} (\bibinfo {year} {2005})}\BibitemShut {NoStop}%
\bibitem [{\citenamefont {Lai}\ and\ \citenamefont {Cheong}(2020{\natexlab{a}})}]{lai_2020_social_review}%
  \BibitemOpen
  \bibfield  {author} {\bibinfo {author} {\bibfnamefont {J.~W.}\ \bibnamefont {Lai}}\ and\ \bibinfo {author} {\bibfnamefont {K.~H.}\ \bibnamefont {Cheong}},\ }\href {https://doi.org/10.1007/s11071-020-05738-9} {\bibfield  {journal} {\bibinfo  {journal} {Nonlinear Dynamics}\ }\textbf {\bibinfo {volume} {101}},\ \bibinfo {pages} {1} (\bibinfo {year} {2020}{\natexlab{a}})}\BibitemShut {NoStop}%
\bibitem [{\citenamefont {Mishra}\ \emph {et~al.}(2024)\citenamefont {Mishra}, \citenamefont {Wen},\ and\ \citenamefont {Cheong}}]{mishra_2024_routing}%
  \BibitemOpen
  \bibfield  {author} {\bibinfo {author} {\bibfnamefont {A.}~\bibnamefont {Mishra}}, \bibinfo {author} {\bibfnamefont {T.}~\bibnamefont {Wen}},\ and\ \bibinfo {author} {\bibfnamefont {K.~H.}\ \bibnamefont {Cheong}},\ }\href {https://doi.org/10.1103/PhysRevResearch.6.L012037} {\bibfield  {journal} {\bibinfo  {journal} {Physical Review Research}\ }\textbf {\bibinfo {volume} {6}},\ \bibinfo {pages} {L012037} (\bibinfo {year} {2024})}\BibitemShut {NoStop}%
\bibitem [{\citenamefont {Pires}\ \emph {et~al.}(2026)\citenamefont {Pires}, \citenamefont {Pinto}, \citenamefont {C{\'a}novas},\ and\ \citenamefont {Queir{\'o}s}}]{pires_2026_chaos}%
  \BibitemOpen
  \bibfield  {author} {\bibinfo {author} {\bibfnamefont {M.~A.}\ \bibnamefont {Pires}}, \bibinfo {author} {\bibfnamefont {E.~P.}\ \bibnamefont {Pinto}}, \bibinfo {author} {\bibfnamefont {J.~S.}\ \bibnamefont {C{\'a}novas}},\ and\ \bibinfo {author} {\bibfnamefont {S.~M.~D.}\ \bibnamefont {Queir{\'o}s}},\ }\bibfield  {journal} {\bibinfo  {journal} {The European Physical Journal Special Topics}\ }\href {https://doi.org/10.1140/epjs/s11734-026-02222-0} {10.1140/epjs/s11734-026-02222-0} (\bibinfo {year} {2026}),\ \Eprint {https://arxiv.org/abs/2602.08135} {arXiv:2602.08135 [nlin]} \BibitemShut {NoStop}%
\bibitem [{\citenamefont {Amengual}\ \emph {et~al.}(2004)\citenamefont {Amengual}, \citenamefont {Allison}, \citenamefont {Toral},\ and\ \citenamefont {Abbott}}]{amengual_2004_discrete}%
  \BibitemOpen
  \bibfield  {author} {\bibinfo {author} {\bibfnamefont {P.}~\bibnamefont {Amengual}}, \bibinfo {author} {\bibfnamefont {A.}~\bibnamefont {Allison}}, \bibinfo {author} {\bibfnamefont {R.}~\bibnamefont {Toral}},\ and\ \bibinfo {author} {\bibfnamefont {D.}~\bibnamefont {Abbott}},\ }\href {https://doi.org/10.1098/rspa.2004.1283} {\bibfield  {journal} {\bibinfo  {journal} {Proceedings of the Royal Society of London. Series A: Mathematical, Physical and Engineering Sciences}\ }\textbf {\bibinfo {volume} {460}},\ \bibinfo {pages} {2269} (\bibinfo {year} {2004})}\BibitemShut {NoStop}%
\bibitem [{\citenamefont {Flitney}(2012)}]{flitney_2012_quantum}%
  \BibitemOpen
  \bibfield  {author} {\bibinfo {author} {\bibfnamefont {A.~P.}\ \bibnamefont {Flitney}},\ }\href {https://doi.org/10.48550/arXiv.1209.2252} {\bibinfo {title} {Quantum {{Parrondo}}'s games using quantum walks}} (\bibinfo {year} {2012}),\ \Eprint {https://arxiv.org/abs/1209.2252} {arXiv:1209.2252 [quant-ph]} \BibitemShut {NoStop}%
\bibitem [{\citenamefont {Jan}\ \emph {et~al.}(2020)\citenamefont {Jan}, \citenamefont {Wang}, \citenamefont {Xu}, \citenamefont {Pan}, \citenamefont {Chen}, \citenamefont {Han}, \citenamefont {Li}, \citenamefont {Guo},\ and\ \citenamefont {Abbott}}]{jan_2020_experimental}%
  \BibitemOpen
  \bibfield  {author} {\bibinfo {author} {\bibfnamefont {M.}~\bibnamefont {Jan}}, \bibinfo {author} {\bibfnamefont {Q.-Q.}\ \bibnamefont {Wang}}, \bibinfo {author} {\bibfnamefont {X.-Y.}\ \bibnamefont {Xu}}, \bibinfo {author} {\bibfnamefont {W.-W.}\ \bibnamefont {Pan}}, \bibinfo {author} {\bibfnamefont {Z.}~\bibnamefont {Chen}}, \bibinfo {author} {\bibfnamefont {Y.-J.}\ \bibnamefont {Han}}, \bibinfo {author} {\bibfnamefont {C.-F.}\ \bibnamefont {Li}}, \bibinfo {author} {\bibfnamefont {G.-C.}\ \bibnamefont {Guo}},\ and\ \bibinfo {author} {\bibfnamefont {D.}~\bibnamefont {Abbott}},\ }\href {https://doi.org/10.1002/qute.201900127} {\bibfield  {journal} {\bibinfo  {journal} {Advanced Quantum Technologies}\ }\textbf {\bibinfo {volume} {3}},\ \bibinfo {pages} {1900127} (\bibinfo {year} {2020})}\BibitemShut {NoStop}%
\bibitem [{\citenamefont {Aharonov}\ \emph {et~al.}(1993)\citenamefont {Aharonov}, \citenamefont {Davidovich},\ and\ \citenamefont {Zagury}}]{Aharonov1993}%
  \BibitemOpen
  \bibfield  {author} {\bibinfo {author} {\bibfnamefont {Y.}~\bibnamefont {Aharonov}}, \bibinfo {author} {\bibfnamefont {L.}~\bibnamefont {Davidovich}},\ and\ \bibinfo {author} {\bibfnamefont {N.}~\bibnamefont {Zagury}},\ }\href {https://doi.org/10.1103/PhysRevA.48.1687} {\bibfield  {journal} {\bibinfo  {journal} {Phys. Rev. A}\ }\textbf {\bibinfo {volume} {48}},\ \bibinfo {pages} {1687} (\bibinfo {year} {1993})}\BibitemShut {NoStop}%
\bibitem [{\citenamefont {Kempe}(2003)}]{Kempe2003}%
  \BibitemOpen
  \bibfield  {author} {\bibinfo {author} {\bibfnamefont {J.}~\bibnamefont {Kempe}},\ }\href {https://doi.org/10.1080/00107151031000110776} {\bibfield  {journal} {\bibinfo  {journal} {Contemporary Physics}\ }\textbf {\bibinfo {volume} {44}},\ \bibinfo {pages} {307} (\bibinfo {year} {2003})}\BibitemShut {NoStop}%
\bibitem [{\citenamefont {Flitney}\ \emph {et~al.}(2004)\citenamefont {Flitney}, \citenamefont {Abbott},\ and\ \citenamefont {Johnson}}]{flitney2004}%
  \BibitemOpen
  \bibfield  {author} {\bibinfo {author} {\bibfnamefont {A.~P.}\ \bibnamefont {Flitney}}, \bibinfo {author} {\bibfnamefont {D.}~\bibnamefont {Abbott}},\ and\ \bibinfo {author} {\bibfnamefont {N.~F.}\ \bibnamefont {Johnson}},\ }\href {https://doi.org/10.1088/0305-4470/37/30/013} {\bibfield  {journal} {\bibinfo  {journal} {Journal of Physics A: Mathematical and General}\ }\textbf {\bibinfo {volume} {37}},\ \bibinfo {pages} {7581} (\bibinfo {year} {2004})}\BibitemShut {NoStop}%
\bibitem [{\citenamefont {Ko{\v{s}}{\'i}k}\ \emph {et~al.}(2007)\citenamefont {Ko{\v{s}}{\'i}k}, \citenamefont {Miszczak},\ and\ \citenamefont {Bu{\v{z}}ek}}]{kosik_2007_quantum}%
  \BibitemOpen
  \bibfield  {author} {\bibinfo {author} {\bibfnamefont {J.}~\bibnamefont {Ko{\v{s}}{\'i}k}}, \bibinfo {author} {\bibfnamefont {J.~A.}\ \bibnamefont {Miszczak}},\ and\ \bibinfo {author} {\bibfnamefont {V.}~\bibnamefont {Bu{\v{z}}ek}},\ }\href {https://doi.org/10.1080/09500340701408722} {\bibfield  {journal} {\bibinfo  {journal} {Journal of Modern Optics}\ }\textbf {\bibinfo {volume} {54}},\ \bibinfo {pages} {2275} (\bibinfo {year} {2007})}\BibitemShut {NoStop}%
\bibitem [{\citenamefont {Bulger}\ \emph {et~al.}(2008)\citenamefont {Bulger}, \citenamefont {Freckleton},\ and\ \citenamefont {Twamley}}]{bulger_2008_positiondependent}%
  \BibitemOpen
  \bibfield  {author} {\bibinfo {author} {\bibfnamefont {D.}~\bibnamefont {Bulger}}, \bibinfo {author} {\bibfnamefont {J.}~\bibnamefont {Freckleton}},\ and\ \bibinfo {author} {\bibfnamefont {J.}~\bibnamefont {Twamley}},\ }\href {https://doi.org/10.1088/1367-2630/10/9/093014} {\bibfield  {journal} {\bibinfo  {journal} {New J. Phys.}\ }\textbf {\bibinfo {volume} {10}},\ \bibinfo {pages} {093014} (\bibinfo {year} {2008})}\BibitemShut {NoStop}%
\bibitem [{\citenamefont {Chandrashekar}\ and\ \citenamefont {Banerjee}(2011)}]{chandrashekar_2011_parrondos}%
  \BibitemOpen
  \bibfield  {author} {\bibinfo {author} {\bibfnamefont {C.~M.}\ \bibnamefont {Chandrashekar}}\ and\ \bibinfo {author} {\bibfnamefont {S.}~\bibnamefont {Banerjee}},\ }\href {https://doi.org/10.1016/j.physleta.2011.02.071} {\bibfield  {journal} {\bibinfo  {journal} {Physics Letters A}\ }\textbf {\bibinfo {volume} {375}},\ \bibinfo {pages} {1553} (\bibinfo {year} {2011})}\BibitemShut {NoStop}%
\bibitem [{\citenamefont {Li}\ \emph {et~al.}(2013)\citenamefont {Li}, \citenamefont {Zhang},\ and\ \citenamefont {Guo}}]{li_2013_quantum}%
  \BibitemOpen
  \bibfield  {author} {\bibinfo {author} {\bibfnamefont {M.}~\bibnamefont {Li}}, \bibinfo {author} {\bibfnamefont {Y.-S.}\ \bibnamefont {Zhang}},\ and\ \bibinfo {author} {\bibfnamefont {G.-C.}\ \bibnamefont {Guo}},\ }\href {https://doi.org/10.1142/S0219477513500247} {\bibfield  {journal} {\bibinfo  {journal} {Fluctuation and Noise Letters}\ }\textbf {\bibinfo {volume} {12}},\ \bibinfo {pages} {1350024} (\bibinfo {year} {2013})}\BibitemShut {NoStop}%
\bibitem [{\citenamefont {Pawela}\ and\ \citenamefont {S{\l}adkowski}(2013)}]{pawela_2013_cooperative}%
  \BibitemOpen
  \bibfield  {author} {\bibinfo {author} {\bibfnamefont {{\L}.}~\bibnamefont {Pawela}}\ and\ \bibinfo {author} {\bibfnamefont {J.}~\bibnamefont {S{\l}adkowski}},\ }\href {https://doi.org/10.1016/j.physd.2013.04.010} {\bibfield  {journal} {\bibinfo  {journal} {Physica D: Nonlinear Phenomena}\ }\textbf {\bibinfo {volume} {256--257}},\ \bibinfo {pages} {51} (\bibinfo {year} {2013})}\BibitemShut {NoStop}%
\bibitem [{\citenamefont {Rajendran}\ and\ \citenamefont {Benjamin}(2018{\natexlab{a}})}]{rajendran_2018_implementing}%
  \BibitemOpen
  \bibfield  {author} {\bibinfo {author} {\bibfnamefont {J.}~\bibnamefont {Rajendran}}\ and\ \bibinfo {author} {\bibfnamefont {C.}~\bibnamefont {Benjamin}},\ }\href {https://doi.org/10.1098/rsos.171599} {\bibfield  {journal} {\bibinfo  {journal} {Royal Society Open Science}\ }\textbf {\bibinfo {volume} {5}},\ \bibinfo {pages} {171599} (\bibinfo {year} {2018}{\natexlab{a}})}\BibitemShut {NoStop}%
\bibitem [{\citenamefont {Rajendran}\ and\ \citenamefont {Benjamin}(2018{\natexlab{b}})}]{rajendran_2018_playing}%
  \BibitemOpen
  \bibfield  {author} {\bibinfo {author} {\bibfnamefont {J.}~\bibnamefont {Rajendran}}\ and\ \bibinfo {author} {\bibfnamefont {C.}~\bibnamefont {Benjamin}},\ }\href {https://doi.org/10.1209/0295-5075/122/40004} {\bibfield  {journal} {\bibinfo  {journal} {EPL}\ }\textbf {\bibinfo {volume} {122}},\ \bibinfo {pages} {40004} (\bibinfo {year} {2018}{\natexlab{b}})}\BibitemShut {NoStop}%
\bibitem [{\citenamefont {Machida}\ and\ \citenamefont {Gr{\"u}nbaum}(2018)}]{machida_2018_limit}%
  \BibitemOpen
  \bibfield  {author} {\bibinfo {author} {\bibfnamefont {T.}~\bibnamefont {Machida}}\ and\ \bibinfo {author} {\bibfnamefont {F.~A.}\ \bibnamefont {Gr{\"u}nbaum}},\ }\href {https://doi.org/10.1007/s11128-018-2009-4} {\bibfield  {journal} {\bibinfo  {journal} {Quantum Information Processing}\ }\textbf {\bibinfo {volume} {17}},\ \bibinfo {pages} {241} (\bibinfo {year} {2018})}\BibitemShut {NoStop}%
\bibitem [{\citenamefont {Lai}\ and\ \citenamefont {Cheong}(2020{\natexlab{b}})}]{lai_2020_parrondo}%
  \BibitemOpen
  \bibfield  {author} {\bibinfo {author} {\bibfnamefont {J.~W.}\ \bibnamefont {Lai}}\ and\ \bibinfo {author} {\bibfnamefont {K.~H.}\ \bibnamefont {Cheong}},\ }\href {https://doi.org/10.1103/PhysRevE.101.052212} {\bibfield  {journal} {\bibinfo  {journal} {Phys. Rev. E}\ }\textbf {\bibinfo {volume} {101}},\ \bibinfo {pages} {052212} (\bibinfo {year} {2020}{\natexlab{b}})}\BibitemShut {NoStop}%
\bibitem [{\citenamefont {Lai}\ \emph {et~al.}(2020)\citenamefont {Lai}, \citenamefont {Tan}, \citenamefont {Lu}, \citenamefont {Yap},\ and\ \citenamefont {Cheong}}]{lai2020parrondo}%
  \BibitemOpen
  \bibfield  {author} {\bibinfo {author} {\bibfnamefont {J.~W.}\ \bibnamefont {Lai}}, \bibinfo {author} {\bibfnamefont {J.~R.~A.}\ \bibnamefont {Tan}}, \bibinfo {author} {\bibfnamefont {H.}~\bibnamefont {Lu}}, \bibinfo {author} {\bibfnamefont {Z.~R.}\ \bibnamefont {Yap}},\ and\ \bibinfo {author} {\bibfnamefont {K.~H.}\ \bibnamefont {Cheong}},\ }\href {https://doi.org/10.1103/PhysRevE.102.012213} {\bibfield  {journal} {\bibinfo  {journal} {Physical Review E}\ }\textbf {\bibinfo {volume} {102}},\ \bibinfo {pages} {012213} (\bibinfo {year} {2020})}\BibitemShut {NoStop}%
\bibitem [{\citenamefont {Pires}\ and\ \citenamefont {Queir{\'o}s}(2020)}]{pires_2020_parrondos}%
  \BibitemOpen
  \bibfield  {author} {\bibinfo {author} {\bibfnamefont {M.~A.}\ \bibnamefont {Pires}}\ and\ \bibinfo {author} {\bibfnamefont {S.~M.~D.}\ \bibnamefont {Queir{\'o}s}},\ }\href {https://doi.org/10.1103/PhysRevE.102.042124} {\bibfield  {journal} {\bibinfo  {journal} {Physical Review E}\ }\textbf {\bibinfo {volume} {102}},\ \bibinfo {pages} {042124} (\bibinfo {year} {2020})}\BibitemShut {NoStop}%
\bibitem [{\citenamefont {Walczak}\ and\ \citenamefont {Bauer}(2021)}]{walczak_2021_parrondos}%
  \BibitemOpen
  \bibfield  {author} {\bibinfo {author} {\bibfnamefont {Z.}~\bibnamefont {Walczak}}\ and\ \bibinfo {author} {\bibfnamefont {J.~H.}\ \bibnamefont {Bauer}},\ }\href {https://doi.org/10.1103/PhysRevE.104.064209} {\bibfield  {journal} {\bibinfo  {journal} {Physical Review E}\ }\textbf {\bibinfo {volume} {104}},\ \bibinfo {pages} {064209} (\bibinfo {year} {2021})}\BibitemShut {NoStop}%
\bibitem [{\citenamefont {Walczak}\ and\ \citenamefont {Bauer}(2022)}]{walczak_2022_parrondos}%
  \BibitemOpen
  \bibfield  {author} {\bibinfo {author} {\bibfnamefont {Z.}~\bibnamefont {Walczak}}\ and\ \bibinfo {author} {\bibfnamefont {J.~H.}\ \bibnamefont {Bauer}},\ }\href {https://doi.org/10.1103/PhysRevE.105.064211} {\bibfield  {journal} {\bibinfo  {journal} {Physical Review E}\ }\textbf {\bibinfo {volume} {105}},\ \bibinfo {pages} {064211} (\bibinfo {year} {2022})}\BibitemShut {NoStop}%
\bibitem [{\citenamefont {Trautmann}\ \emph {et~al.}(2022)\citenamefont {Trautmann}, \citenamefont {Groiseau},\ and\ \citenamefont {Wimberger}}]{trautmann_2022_momentum}%
  \BibitemOpen
  \bibfield  {author} {\bibinfo {author} {\bibfnamefont {G.}~\bibnamefont {Trautmann}}, \bibinfo {author} {\bibfnamefont {C.}~\bibnamefont {Groiseau}},\ and\ \bibinfo {author} {\bibfnamefont {S.}~\bibnamefont {Wimberger}},\ }\href {https://doi.org/10.1142/S0219477522500535} {\bibfield  {journal} {\bibinfo  {journal} {Fluctuation and Noise Letters}\ }\textbf {\bibinfo {volume} {21}},\ \bibinfo {pages} {2250053} (\bibinfo {year} {2022})}\BibitemShut {NoStop}%
\bibitem [{\citenamefont {Jan}\ \emph {et~al.}(2023)\citenamefont {Jan}, \citenamefont {Khan},\ and\ \citenamefont {Xianlong}}]{jan_2023_territories}%
  \BibitemOpen
  \bibfield  {author} {\bibinfo {author} {\bibfnamefont {M.}~\bibnamefont {Jan}}, \bibinfo {author} {\bibfnamefont {N.~A.}\ \bibnamefont {Khan}},\ and\ \bibinfo {author} {\bibfnamefont {G.}~\bibnamefont {Xianlong}},\ }\href {https://doi.org/10.1140/epjp/s13360-023-03685-z} {\bibfield  {journal} {\bibinfo  {journal} {Eur. Phys. J. Plus}\ }\textbf {\bibinfo {volume} {138}},\ \bibinfo {pages} {1} (\bibinfo {year} {2023})}\BibitemShut {NoStop}%
\bibitem [{\citenamefont {Mielke}(2023)}]{mielke_2023_lowdimensional}%
  \BibitemOpen
  \bibfield  {author} {\bibinfo {author} {\bibfnamefont {A.}~\bibnamefont {Mielke}},\ }\href {https://doi.org/10.48550/arXiv.2306.16845} {\bibinfo {title} {Quantum {{Parrondo Games}} in {{Low-Dimensional Hilbert Spaces}}}} (\bibinfo {year} {2023}),\ \Eprint {https://arxiv.org/abs/2306.16845} {arXiv:2306.16845 [quant-ph]} \BibitemShut {NoStop}%
\bibitem [{\citenamefont {Walczak}\ and\ \citenamefont {Bauer}(2023)}]{walczak_2023_noiseinduced}%
  \BibitemOpen
  \bibfield  {author} {\bibinfo {author} {\bibfnamefont {Z.}~\bibnamefont {Walczak}}\ and\ \bibinfo {author} {\bibfnamefont {J.~H.}\ \bibnamefont {Bauer}},\ }\href {https://doi.org/10.1103/PhysRevE.108.044212} {\bibfield  {journal} {\bibinfo  {journal} {Physical Review E}\ }\textbf {\bibinfo {volume} {108}},\ \bibinfo {pages} {044212} (\bibinfo {year} {2023})}\BibitemShut {NoStop}%
\bibitem [{\citenamefont {Ximenes}\ \emph {et~al.}(2024)\citenamefont {Ximenes}, \citenamefont {Pires},\ and\ \citenamefont {{Villas-B{\^o}as}}}]{ximenes2024parrondo}%
  \BibitemOpen
  \bibfield  {author} {\bibinfo {author} {\bibfnamefont {J.~J.}\ \bibnamefont {Ximenes}}, \bibinfo {author} {\bibfnamefont {M.~A.}\ \bibnamefont {Pires}},\ and\ \bibinfo {author} {\bibfnamefont {J.~M.}\ \bibnamefont {{Villas-B{\^o}as}}},\ }\href {https://doi.org/10.1103/PhysRevA.109.032417} {\bibfield  {journal} {\bibinfo  {journal} {Phys. Rev. A}\ }\textbf {\bibinfo {volume} {109}},\ \bibinfo {pages} {032417} (\bibinfo {year} {2024})}\BibitemShut {NoStop}%
\bibitem [{\citenamefont {Kadiri}(2024)}]{kadiri_2024_scouring}%
  \BibitemOpen
  \bibfield  {author} {\bibinfo {author} {\bibfnamefont {G.}~\bibnamefont {Kadiri}},\ }\href {https://doi.org/10.1103/PhysRevA.110.022421} {\bibfield  {journal} {\bibinfo  {journal} {Physical Review A}\ }\textbf {\bibinfo {volume} {110}},\ \bibinfo {pages} {022421} (\bibinfo {year} {2024})}\BibitemShut {NoStop}%
\bibitem [{\citenamefont {Mittal}\ and\ \citenamefont {Huang}(2024)}]{mittal_2024_parrondos}%
  \BibitemOpen
  \bibfield  {author} {\bibinfo {author} {\bibfnamefont {V.}~\bibnamefont {Mittal}}\ and\ \bibinfo {author} {\bibfnamefont {Y.-P.}\ \bibnamefont {Huang}},\ }\href {https://doi.org/10.1103/PhysRevA.110.052440} {\bibfield  {journal} {\bibinfo  {journal} {Physical Review A}\ }\textbf {\bibinfo {volume} {110}},\ \bibinfo {pages} {052440} (\bibinfo {year} {2024})},\ \Eprint {https://arxiv.org/abs/2407.16558} {arXiv:2407.16558 [quant-ph]} \BibitemShut {NoStop}%
\bibitem [{\citenamefont {Walczak}\ and\ \citenamefont {Bauer}(2024)}]{walczak2024parrondo}%
  \BibitemOpen
  \bibfield  {author} {\bibinfo {author} {\bibfnamefont {Z.}~\bibnamefont {Walczak}}\ and\ \bibinfo {author} {\bibfnamefont {J.~H.}\ \bibnamefont {Bauer}},\ }\href {https://doi.org/10.1007/s11128-024-04614-4} {\bibfield  {journal} {\bibinfo  {journal} {Quantum Information Processing}\ }\textbf {\bibinfo {volume} {23}},\ \bibinfo {pages} {408} (\bibinfo {year} {2024})}\BibitemShut {NoStop}%
\bibitem [{\citenamefont {Walczak}\ and\ \citenamefont {Bauer}(2025)}]{walczak_2025_parrondos}%
  \BibitemOpen
  \bibfield  {author} {\bibinfo {author} {\bibfnamefont {Z.}~\bibnamefont {Walczak}}\ and\ \bibinfo {author} {\bibfnamefont {J.~H.}\ \bibnamefont {Bauer}},\ }\href {https://doi.org/10.1103/k6nn-n6c4} {\bibfield  {journal} {\bibinfo  {journal} {Physical Review E}\ }\textbf {\bibinfo {volume} {111}},\ \bibinfo {pages} {064218} (\bibinfo {year} {2025})}\BibitemShut {NoStop}%
\bibitem [{\citenamefont {Chen}\ \emph {et~al.}(2025)\citenamefont {Chen}, \citenamefont {Chang},\ and\ \citenamefont {Huang}}]{chen_2025_singlequbit}%
  \BibitemOpen
  \bibfield  {author} {\bibinfo {author} {\bibfnamefont {Y.-H.}\ \bibnamefont {Chen}}, \bibinfo {author} {\bibfnamefont {R.-Y.}\ \bibnamefont {Chang}},\ and\ \bibinfo {author} {\bibfnamefont {T.-W.}\ \bibnamefont {Huang}},\ }in\ \href {https://doi.org/10.1109/QCE65121.2025.10396} {\emph {\bibinfo {booktitle} {2025 IEEE International Conference on Quantum Computing and Engineering (QCE)}}},\ Vol.~\bibinfo {volume} {2}\ (\bibinfo {year} {2025})\ pp.\ \bibinfo {pages} {462--463}\BibitemShut {NoStop}%
\bibitem [{\citenamefont {Ximenes}\ \emph {et~al.}(2026)\citenamefont {Ximenes}, \citenamefont {Pires},\ and\ \citenamefont {Villas-B{\^o}as}}]{ximenes_2026_aperiodic}%
  \BibitemOpen
  \bibfield  {author} {\bibinfo {author} {\bibfnamefont {J.~J.}\ \bibnamefont {Ximenes}}, \bibinfo {author} {\bibfnamefont {M.~A.}\ \bibnamefont {Pires}},\ and\ \bibinfo {author} {\bibfnamefont {J.~M.}\ \bibnamefont {Villas-B{\^o}as}},\ }\href {https://doi.org/10.1103/r2xh-yyns} {\bibfield  {journal} {\bibinfo  {journal} {Physical Review A}\ }\textbf {\bibinfo {volume} {113}},\ \bibinfo {pages} {022410} (\bibinfo {year} {2026})}\BibitemShut {NoStop}%
\bibitem [{\citenamefont {Cordeiro}\ \emph {et~al.}(2026)\citenamefont {Cordeiro}, \citenamefont {Duzzioni},\ and\ \citenamefont {Amorim}}]{cordeiro_2026_defective}%
  \BibitemOpen
  \bibfield  {author} {\bibinfo {author} {\bibfnamefont {J.~V.}\ \bibnamefont {Cordeiro}}, \bibinfo {author} {\bibfnamefont {E.~I.}\ \bibnamefont {Duzzioni}},\ and\ \bibinfo {author} {\bibfnamefont {E.~P.~M.}\ \bibnamefont {Amorim}},\ }\href {https://doi.org/10.1103/p4pl-rvdh} {\bibfield  {journal} {\bibinfo  {journal} {Physical Review A}\ }\textbf {\bibinfo {volume} {114}},\ \bibinfo {pages} {012459} (\bibinfo {year} {2026})}\BibitemShut {NoStop}%
\bibitem [{\citenamefont {Ambainis}\ \emph {et~al.}(2001)\citenamefont {Ambainis}, \citenamefont {Bach}, \citenamefont {Nayak}, \citenamefont {Vishwanath},\ and\ \citenamefont {Watrous}}]{ambainis_2001_onedimensional}%
  \BibitemOpen
  \bibfield  {author} {\bibinfo {author} {\bibfnamefont {A.}~\bibnamefont {Ambainis}}, \bibinfo {author} {\bibfnamefont {E.}~\bibnamefont {Bach}}, \bibinfo {author} {\bibfnamefont {A.}~\bibnamefont {Nayak}}, \bibinfo {author} {\bibfnamefont {A.}~\bibnamefont {Vishwanath}},\ and\ \bibinfo {author} {\bibfnamefont {J.}~\bibnamefont {Watrous}},\ }in\ \href {https://doi.org/10.1145/380752.380757} {\emph {\bibinfo {booktitle} {Proceedings of the Thirty-Third Annual {{ACM}} Symposium on {{Theory}} of Computing}}}\ (\bibinfo  {publisher} {ACM},\ \bibinfo {address} {Hersonissos Greece},\ \bibinfo {year} {2001})\ pp.\ \bibinfo {pages} {37--49}\BibitemShut {NoStop}%
\bibitem [{\citenamefont {Konno}(2002)}]{konno_2002_quantum}%
  \BibitemOpen
  \bibfield  {author} {\bibinfo {author} {\bibfnamefont {N.}~\bibnamefont {Konno}},\ }\href {https://doi.org/10.1023/A:1023413713008} {\bibfield  {journal} {\bibinfo  {journal} {Quantum Information Processing}\ }\textbf {\bibinfo {volume} {1}},\ \bibinfo {pages} {345} (\bibinfo {year} {2002})}\BibitemShut {NoStop}%
\bibitem [{\citenamefont {Grimmett}\ \emph {et~al.}(2004)\citenamefont {Grimmett}, \citenamefont {Janson},\ and\ \citenamefont {Scudo}}]{grimmett_2004_weak}%
  \BibitemOpen
  \bibfield  {author} {\bibinfo {author} {\bibfnamefont {G.}~\bibnamefont {Grimmett}}, \bibinfo {author} {\bibfnamefont {S.}~\bibnamefont {Janson}},\ and\ \bibinfo {author} {\bibfnamefont {P.~F.}\ \bibnamefont {Scudo}},\ }\href {https://doi.org/10.1103/PhysRevE.69.026119} {\bibfield  {journal} {\bibinfo  {journal} {Physical Review E}\ }\textbf {\bibinfo {volume} {69}},\ \bibinfo {pages} {026119} (\bibinfo {year} {2004})}\BibitemShut {NoStop}%
\bibitem [{\citenamefont {Konno}(2005)}]{konno_2005_limit}%
  \BibitemOpen
  \bibfield  {author} {\bibinfo {author} {\bibfnamefont {N.}~\bibnamefont {Konno}},\ }\href {https://doi.org/10.2969/jmsj/1150287309} {\bibfield  {journal} {\bibinfo  {journal} {Journal of the Mathematical Society of Japan}\ }\textbf {\bibinfo {volume} {57}},\ \bibinfo {pages} {1179} (\bibinfo {year} {2005})}\BibitemShut {NoStop}%
\bibitem [{\citenamefont {Ahlbrecht}\ \emph {et~al.}(2011)\citenamefont {Ahlbrecht}, \citenamefont {Vogts}, \citenamefont {Werner},\ and\ \citenamefont {Werner}}]{ahlbrecht_2011_asymptotic}%
  \BibitemOpen
  \bibfield  {author} {\bibinfo {author} {\bibfnamefont {A.}~\bibnamefont {Ahlbrecht}}, \bibinfo {author} {\bibfnamefont {H.}~\bibnamefont {Vogts}}, \bibinfo {author} {\bibfnamefont {A.~H.}\ \bibnamefont {Werner}},\ and\ \bibinfo {author} {\bibfnamefont {R.~F.}\ \bibnamefont {Werner}},\ }\href {https://doi.org/10.1063/1.3575568} {\bibfield  {journal} {\bibinfo  {journal} {Journal of Mathematical Physics}\ }\textbf {\bibinfo {volume} {52}},\ \bibinfo {pages} {042201} (\bibinfo {year} {2011})}\BibitemShut {NoStop}%
\bibitem [{\citenamefont {Annabestani}(2020)}]{annabestani_2020_asymptotic}%
  \BibitemOpen
  \bibfield  {author} {\bibinfo {author} {\bibfnamefont {M.}~\bibnamefont {Annabestani}},\ }\href {https://doi.org/10.1088/1751-8121/ab7b20} {\bibfield  {journal} {\bibinfo  {journal} {Journal of Physics A: Mathematical and Theoretical}\ }\textbf {\bibinfo {volume} {53}},\ \bibinfo {pages} {155303} (\bibinfo {year} {2020})}\BibitemShut {NoStop}%
\bibitem [{\citenamefont {Portugal}(2018)}]{portugal_2018_quantum}%
  \BibitemOpen
  \bibfield  {author} {\bibinfo {author} {\bibfnamefont {R.}~\bibnamefont {Portugal}},\ }\href {https://doi.org/10.1007/978-3-319-97813-0} {\emph {\bibinfo {title} {Quantum {{Walks}} and {{Search Algorithms}}}}},\ Quantum {{Science}} and {{Technology}}\ (\bibinfo  {publisher} {Springer International Publishing},\ \bibinfo {address} {Cham},\ \bibinfo {year} {2018})\BibitemShut {NoStop}%
\bibitem [{\citenamefont {Nayak}\ and\ \citenamefont {Vishwanath}(2000)}]{nayak_2000_quantum}%
  \BibitemOpen
  \bibfield  {author} {\bibinfo {author} {\bibfnamefont {A.}~\bibnamefont {Nayak}}\ and\ \bibinfo {author} {\bibfnamefont {A.}~\bibnamefont {Vishwanath}},\ }\href {https://doi.org/10.48550/arXiv.quant-ph/0010117} {\bibinfo {title} {Quantum {{Walk}} on the {{Line}}}} (\bibinfo {year} {2000}),\ \Eprint {https://arxiv.org/abs/quant-ph/0010117} {arXiv:quant-ph/0010117} \BibitemShut {NoStop}%
\bibitem [{\citenamefont {Ambainis}(2003)}]{ambainis_2003_quantum}%
  \BibitemOpen
  \bibfield  {author} {\bibinfo {author} {\bibfnamefont {A.}~\bibnamefont {Ambainis}},\ }\href {https://doi.org/10.1142/S0219749903000383} {\bibfield  {journal} {\bibinfo  {journal} {International Journal of Quantum Information}\ }\textbf {\bibinfo {volume} {1}},\ \bibinfo {pages} {507} (\bibinfo {year} {2003})}\BibitemShut {NoStop}%
\bibitem [{SM()}]{SM}%
  \BibitemOpen
  \href@noop {} {}\bibinfo {note} {See Supplemental Material for additional derivations and technical details.}\BibitemShut {Stop}%
\bibitem [{Note1()}]{Note1}%
  \BibitemOpen
  \bibinfo {note} {As a finitely generated cone, $\protect \mathcal K$ is closed and convex. Since $\protect \vec {\protect \mathcal {T}}(U_{\protect \mathcal W})\protect \notin \protect \mathcal K$, the separating-hyperplane theorem~\cite {rockafellar_1970_convex} provides a vector $\protect \vec m$ such that $\protect \vec m\cdot \protect \vec {\protect \mathcal {T}}(U_{\protect \mathcal W})>0$ and $\protect \vec m\cdot \protect \vec v\leq 0$ for every $\protect \vec v\in \protect \mathcal K$. Because $(\protect \vec {\protect \mathcal {T}}_i)_z>0$ for every generator, setting $\protect \vec m_\epsilon =\protect \vec m-\epsilon \protect \hat z$ for sufficiently small $\epsilon >0$ gives $\protect \vec m_\epsilon \cdot \protect \vec {\protect \mathcal {T}}_i<0$ for every $i$ while preserving $\protect \vec m_\epsilon \cdot \protect \vec {\protect \mathcal {T}}(U_{\protect \mathcal W})>0$. Normalizing $\protect \vec m_\epsilon $ gives the unit normal $\protect \hat n$.}\BibitemShut {Stop}%
\bibitem [{\citenamefont {Albertini}\ and\ \citenamefont {D'Alessandro}(2012)}]{albertini_2012_controllability}%
  \BibitemOpen
  \bibfield  {author} {\bibinfo {author} {\bibfnamefont {F.}~\bibnamefont {Albertini}}\ and\ \bibinfo {author} {\bibfnamefont {D.}~\bibnamefont {D'Alessandro}},\ }\href {https://doi.org/10.1007/s00498-012-0084-0} {\bibfield  {journal} {\bibinfo  {journal} {Mathematics of Control, Signals, and Systems}\ }\textbf {\bibinfo {volume} {24}},\ \bibinfo {pages} {321} (\bibinfo {year} {2012})}\BibitemShut {NoStop}%
\bibitem [{\citenamefont {Rockafellar}(1970)}]{rockafellar_1970_convex}%
  \BibitemOpen
  \bibfield  {author} {\bibinfo {author} {\bibfnamefont {R.~T.}\ \bibnamefont {Rockafellar}},\ }\href {https://doi.org/10.1515/9781400873173} {\emph {\bibinfo {title} {Convex Analysis}}},\ \bibinfo {series} {Princeton Mathematical Series}\ No.~\bibinfo {number} {28}\ (\bibinfo  {publisher} {Princeton University Press},\ \bibinfo {address} {Princeton, NJ},\ \bibinfo {year} {1970})\BibitemShut {NoStop}%
\bibitem [{\citenamefont {Hantzko}\ and\ \citenamefont {Binkowski}(2026)}]{hantzko_2026_classification}%
  \BibitemOpen
  \bibfield  {author} {\bibinfo {author} {\bibfnamefont {L.}~\bibnamefont {Hantzko}}\ and\ \bibinfo {author} {\bibfnamefont {L.}~\bibnamefont {Binkowski}},\ }\href {https://doi.org/10.1140/epjd/s10053-026-01217-9} {\bibfield  {journal} {\bibinfo  {journal} {The European Physical Journal D}\ }\textbf {\bibinfo {volume} {80}},\ \bibinfo {pages} {112} (\bibinfo {year} {2026})}\BibitemShut {NoStop}%
\bibitem [{\citenamefont {Bengtsson}\ and\ \citenamefont {Zyczkowski}(2006)}]{bengtsson_2006_geometry}%
  \BibitemOpen
  \bibfield  {author} {\bibinfo {author} {\bibfnamefont {I.}~\bibnamefont {Bengtsson}}\ and\ \bibinfo {author} {\bibfnamefont {K.}~\bibnamefont {Zyczkowski}},\ }\href {https://doi.org/10.1017/CBO9780511535048} {\emph {\bibinfo {title} {Geometry of {{Quantum States}}: {{An Introduction}} to {{Quantum Entanglement}}}}}\ (\bibinfo  {publisher} {Cambridge University Press},\ \bibinfo {address} {Cambridge},\ \bibinfo {year} {2006})\BibitemShut {NoStop}%
\end{thebibliography}%

% !TeX root = main.tex
% !TeX spellcheck = en_US
\clearpage
\onecolumngrid
\begin{center} \textbf{\large Supplemental Material: \titlepaper}\\[.2cm] \end{center}
% Set stuff {{{*
% Reiniciar contadores
\setcounter{equation}{0}
\setcounter{figure}{0}
\setcounter{table}{0}
\setcounter{section}{0}
\setcounter{secnumdepth}{3}
\setcounter{page}{1}

% Redefinir la numeración para que incluya "S"
\renewcommand{\theequation}{S\arabic{equation}}
\renewcommand{\thefigure}{S\arabic{figure}}
\renewcommand{\thetable}{S\arabic{table}}
\renewcommand{\thesection}{S\arabic{section}}
\renewcommand{\thetheorem}{S\arabic{theorem}}
\renewcommand{\theproof}{S\arabic{proof}}

% Opcional: Redefinir el formato de las referencias si el SM tiene su propia bibliografía
% \renewcommand{\bibnumfmt}[1]{[S#1]}
% \renewcommand{\citenumfont}[1]{S#1}

% *}}}
\section{Asymptotic velocity in terms of the transport vector} % {{{*
\label{sec:asymptotic_velocity_derivation}

To derive the asymptotic velocity $\overline{v}$ of the walker in a one-dimensional discrete-time quantum walk (1D DTQW), let us first move to the quasi-momentum representation. Consider the unitary evolution operator $U$ acting on $\mcH=\mcH_{p}\otimes\mbC^d$, where $\mcH_{p}=\operatorname{span}\{\ket{x}\}_{x\in\mbZ}$. The lattice translation operator $D$ acts as $D(\ket{x}\otimes\ket{c})=\ket{x+1}\otimes\ket{c}$ for any internal state $\ket{c}\in\mbC^d$. If $U$ is translationally invariant, i.e., $[U,D]=0$, the Fourier transform block-diagonalizes it as
\begin{equation}
    U=\int_{-\pi}^{\pi}\frac{dk}{2\pi}\,
    \dyad{k}\otimes\tilde{U}(k),
\end{equation}
where $\ket{k}=\sum_{x\in\mbZ}e^{ikx}\ket{x}$ denotes the Fourier basis, with
$k\in[-\pi,\pi)$ the quasi-momentum~\cite{portugal_2018_quantum}. For each $k$,
$\tilde{U}(k)$ is a $d\times d$ unitary matrix acting on the internal space
$\mbC^d$. For the DTQWs considered in this work, the coin is two-dimensional,
$d=2$, and $\tilde{U}(k)$ is either of the form $\tilde{U}(k)=S(k)C_N\ldots
S(k)C_1$ or $\tilde{U}(k)=S(k)C_N\ldots C_1$, where $S(k)=e^{-ik\sigma_z}$ and
% OLD: the coin operators $C_j$ are $k$-independent\cpnote{ufff como por que? argumentar o ver si es obvio}. Since $\det S(k)=1$, in either
% OLD: case $\det\tilde{U}(k)=\prod_j\det C_j$, which is independent of $k$\cpnote{Para que nos sirve esto?}. The
the coin operators $C_j$ are $k$-independent.
Since $\det S(k)=1$, in either case $\det\tilde{U}(k)=\prod_j\det C_j$, which
is independent of $k$; we use this below to show that the global phase
$\alpha(k)$ is $k$-independent.
% \clnote{Answers the two \cpnote's that were
% here. Added the homogeneity motivation for $C_j$ being $k$-independent, and
% a forward-pointer explaining why $\det\tilde U(k)$ being $k$-independent
% matters (it's used to show $\partial_k\alpha(k)=0$, which is what lets
% Eq.~\eqref{eq:exact_general_position} omit an extra linear-in-$t$ term). Old
% text with the two open cpnotes kept as a comment above. I also added a
% sentence near Eq.~\eqref{eq:exact_general_position} making that connection
% explicit --- see the clnote there. Discuss with JA.} 
The matrix $\tilde{U}(k)\in\mathrm{U}(2)$ can be decomposed into a $\mathrm{U}(1)$
phase factor and an $\mathrm{SU}(2)$ rotation operator:
\begin{equation}\label{eq:Uk}
  \tilde{U}(k)
  =
  e^{-i\alpha(k)}
  \left[
    \cos\omega(k)\mathbb{I}
    -
    i\sin\omega(k)\hat{n}(k)\cdot\vec{\sigma}
  \right]
  =
  e^{-i\alpha(k)}
  e^{-i\omega(k)\hat{n}(k)\cdot\vec{\sigma}}.
\end{equation}
Here, $\alpha(k)$ is a global phase, 
while $2\omega(k)$ and $\hat{n}(k)$ are the rotation angle and
rotation axis of the $\mathrm{SU}(2)$ part of $\tilde U(k)$, respectively.
Throughout this derivation we assume $\sin\omega(k)\neq0$ for every
quasi-momentum $k\in[-\pi,\pi)$, so that $\hat{n}(k)$ and the eigenbasis
$\ket{\phi_\pm(k)}$ remain well defined; this excludes diagonal coin operators
($C_{01}=0$), whose eigenvalues become degenerate at isolated $k$ and which
produce non-generic ballistic-split dynamics (Sec.~\ref{chp:critical_angle_konno}).
The final closed-form transport vector remains finite and correct in this
case too, as confirmed independently via the residue-theorem method of
Sec.~\ref{sm:residue_method}, which does not require pointwise
smoothness of $\hat{n}(k)$. Since the $\mathrm{SU}(2)$ factor has unit determinant,
$\det\tilde{U}(k)=e^{-2i\alpha(k)}$, and therefore
\begin{equation}
\begin{split}
    \alpha(k)
    &=
    -\frac{1}{2}\arg\det\tilde{U}(k)
    \quad (\mathrm{mod}\ \pi) \\
    &=
    -\frac{1}{2}\arg\left(\prod_j\det C_j\right)
    \quad (\mathrm{mod}\ \pi).
\end{split}
\end{equation}
% OLD: Thus, for the DTQWs considered here, $\alpha(k)$ can be chosen \cpnote{porqeu
% OLD: chosen, no es una constante independiente de k?}as a $k$-independent constant
% OLD: $\alpha$, so that $\partial_k\alpha(k)=0$.
Thus, for the DTQWs considered here, since $\det\tilde{U}(k)$ is itself
independent of $k$, continuity in $k$ fixes the branch of the $(\mathrm{mod}\
\pi)$ ambiguity above, so $\alpha(k)$ is a genuine $k$-independent constant
$\alpha$, with $\partial_k\alpha(k)=0$.
% \clnote{Reworded per the ``porque
% chosen'' cpnote: ``can be chosen'' made it sound like a free choice at each
% $k$, when really $\det\tilde U(k)$'s $k$-independence pins the branch once
% and for all by continuity, making $\alpha$ a determined constant, not a
% choice. Old text kept as a comment above. Discuss with JA.}
% OLD: By the spectral theorem,
% OLD: $\tilde{U}(k)$ is diagonalized by an orthonormal eigenbasis
% OLD: $\ket{\phi_\pm(k)}$, with corresponding eigenvalues
% OLD: $e^{-i[\alpha\pm\omega(k)]}$\cpnote{Esto está como desconectado}.
To evolve an arbitrary coin state under $\tilde{U}(k)$, we diagonalize it: by
the spectral theorem, $\tilde{U}(k)$ admits an orthonormal eigenbasis
$\ket{\phi_\pm(k)}$, with eigenvalues $e^{-i[\alpha\pm\omega(k)]}$, which we
use below to expand the initial coin state.
% \clnote{Added a linking clause
% per the ``desconectado'' cpnote --- the spectral-theorem sentence previously
% appeared with no stated purpose; now it points forward to its use a few
% lines below, where $\ket{c_0}$ is expanded in this eigenbasis. Old text kept
% as a comment above. Discuss with JA.}

We consider a walker initially localized at the origin with coin state
$\ket{c_0}$. In position space, the initial state is
$\ket{\psi_0}=\ket{0}\otimes\ket{c_0}$.
% OLD: Since
% OLD: $\ket{0}=\int_{-\pi}^{\pi}\frac{dk}{2\pi}\ket{k}$ the initial state\cpnote{por?}, in the
% OLD: quasi-momentum representation, is $\ket{\tilde\psi_0(k)}=\ket{c_0}$, for every
% OLD: $k$.
Using $\ket{0}=\int_{-\pi}^{\pi}\frac{dk}{2\pi}\ket{k}$, this reads
$\ket{\psi_0}=\int_{-\pi}^{\pi}\frac{dk}{2\pi}\ket{k}\otimes\ket{c_0}$;
comparing term by term with the general quasi-momentum expansion
$\ket{\psi_0}=\int_{-\pi}^{\pi}\frac{dk}{2\pi}\ket{k}\otimes\ket{\tilde\psi_0(k)}$
gives $\ket{\tilde\psi_0(k)}=\ket{c_0}$ for every
$k$.
% \clnote{Spelled out the comparison step per the ``por?'' cpnote, instead
% of asserting $\ket{\tilde\psi_0(k)}=\ket{c_0}$ directly. Old text kept as a
% comment above. Discuss with JA.} 
For each quasi-momentum $k$, we expand the same initial coin state
$\ket{c_0}$ in the $k$-dependent eigenbasis of $\tilde U(k)$:
\begin{equation}
  \ket{c_0} = 
  c_+(k)\ket{\phi_+(k)} + 
  c_-(k)\ket{\phi_-(k)},
\end{equation}
where $c_\pm(k)= \braket{\phi_\pm(k)}{c_0}$. Although $\ket{c_0}$ is independent of $k$, its expansion coefficients depend on $k$ because the eigenbasis  $\qty{\ket{\phi_\pm(k)}}$ varies with quasi-momentum. The quasi-momentum component at an arbitrary discrete time step $t$ is obtained by applying $\tilde U(k)$ $t$ times:
\begin{equation}\label{eq:momentum_state_evolution}
    \ket{\tilde\psi_t(k)} =
    \qty[\tilde U(k)]^t \ket{\tilde\psi_0(k)} =
    c_+(k)e^{-i[\alpha+\omega(k)]t}\ket{\phi_+(k)} +
    c_-(k)e^{-i[\alpha-\omega(k)]t}\ket{\phi_-(k)}.
\end{equation}

In the quasi-momentum representation, the expectation value of the position operator at time $t$ is $\expval{x}_t = \int_{-\pi}^{\pi} \frac{dk}{2\pi} \bra{\tilde\psi_t(k)} i \partial_k \ket{\tilde\psi_t(k)}$. Evaluating this inner product explicitly yields three contributions:
\begin{align}\label{eq:exact_general_position}
    \expval{x}_t = \int_{-\pi}^{\pi} \frac{dk}{2\pi} \Bigg\{ 
    & t\,\partial_k\omega(k)
    \left(
    |c_+(k)|^2-|c_-(k)|^2
    \right) \nonumber \\
    & + \sum_{s=\pm} \left( i c_s^*(k) \partial_k c_s(k)- |c_s(k)|^2 \mathcal{A}_{ss}(k) \right) \nonumber \\
    & - e^{2i\omega(k)t} c_+^*(k) c_-(k)\mathcal{A}_{+-}(k) - e^{-2i\omega(k)t} c_-^*(k) c_+(k)\mathcal{A}_{-+}(k) \Bigg\},
\end{align}
where $\mathcal{A}_{rs}(k) = -i \mel{\phi_r(k)}{\partial_k}{\phi_s(k)}$. No
term proportional to $\partial_k\alpha(k)$ appears because $\alpha(k)$ was
shown above to be $k$-independent. The
first line of \Eref{eq:exact_general_position} grows linearly with time,
whereas the remaining terms contain no factor that grows with $t$. Therefore,
dividing by $t$ and taking the long-time limit leaves only the first
contribution:
% OLD: $\overline{v}=\lim_{t\to\infty}\expval{x}_t/t$\cpnote{esta ecuacion la dejaria inline en la ecuación que viene}, leaves only the first
% OLD: contribution:
% OLD: \begin{equation}\label{eq:vel_almost}
% OLD:     \overline{v}
% OLD:     =
% OLD:     \int_{-\pi}^{\pi}\frac{dk}{2\pi}
% OLD:     \partial_k\omega(k) \qty(\qty|c_+(k)|^2-\qty|c_-(k)|^2).
% OLD: \end{equation}
\begin{equation}\label{eq:vel_almost}
    \overline{v}
    =
    \lim_{t\to\infty}\frac{\expval{x}_t}{t}
    =
    \int_{-\pi}^{\pi}\frac{dk}{2\pi}
    \partial_k\omega(k) \qty(\qty|c_+(k)|^2-\qty|c_-(k)|^2).
\end{equation}
The population difference can be written as
\begin{equation}
    \qty|c_+(k)|^2-\qty|c_-(k)|^2
    =
    \ev{\qty(\dyad{\phi_+(k)}-\dyad{\phi_-(k)})}{c_0}
    =
    \bra{c_0}\hat{n}(k)\cdot\vec{\sigma}\ket{c_0}.
\end{equation}
Writing the density matrix of the initial coin state as
$\rho_0=\frac{1}{2}\left(\mathbb{I}+\vec{r}_0\cdot\vec{\sigma}\right)$,
this expectation value becomes
$\bra{c_0}\hat{n}(k)\cdot\vec{\sigma}\ket{c_0}
=\hat{n}(k)\cdot\vec{r}_0$.

Substituting this result into \Eref{eq:vel_almost} and using the fact that $\vec{r}_0$ is independent of the quasi-momentum $k$, the asymptotic velocity takes the form
\begin{equation}\label{eq:T_definition}
  \overline{v}
  =
  \vmcT\cdot\vec{r}_0,
  \qquad
  \vmcT
  =
  \int_{-\pi}^{\pi} \frac{dk}{2\pi} \pdv{\omega(k)}{k} \hat{n}(k).
\end{equation}
Thus, the dependence on the initial coin state is entirely contained in its
Bloch vector $\vec{r}_0$, while the transport vector $\vmcT$ depends only on
the evolution operator of the walk. If $\alpha(k)$ depended on $k$,
\Eref{eq:exact_general_position} would acquire an additional term
$t\,\partial_k\alpha(k)$, yielding
$\overline{v}=\int_{-\pi}^{\pi}\frac{dk}{2\pi}\partial_k\alpha(k)
+\vmcT\cdot\vec{r}_0$. This additional drift is independent of the initial
coin state and vanishes for the DTQWs considered here because
$\partial_k\alpha(k)=0$.
% *}}}
\section{Transport vector of a single-step DTQW} %{{{*
\label{sm:transport_vector_integral}

We first express the transport vector directly in terms of the evolution
operator for the two-state DTQWs considered here, whose determinant is
independent of $k$. We then specialize the result to the single-step walk
and evaluate the integral explicitly.
% $\tilde U(k)=S(k)C$, with $C\in\mathrm{SU}(2)$\janote{this can be omitted, and we can keep the appendix as general until the final part}, 
As shown in Sec.~\ref{sec:asymptotic_velocity_derivation}, the
$\mathrm U(1)$ phase is $k$-independent for the class of DTQWs considered
in this work, and can therefore be omitted. Henceforth,
% \codexnote{
% in this appendix
% }{
throughout the remainder of the Supplemental Material,
% },
$\tilde U(k)$ denotes the corresponding determinant-one representative.

After omitting the global phase, \Eref{eq:Uk} becomes
$\tilde U(k)=\cos\omega(k)\mathbb I
-i\sin\omega(k)\hat n(k)\cdot\vec\sigma$. Taking the trace and the three
traces $\Tr[\sigma_j\tilde U(k)]$, $j=x,y,z$, gives
$\Tr[\tilde U(k)]=2\cos\omega(k)$ and
$\Tr[\vec\sigma\tilde U(k)]=-2i\sin\omega(k)\hat n(k)$.
Differentiating the first identity gives
$\partial_k\omega(k)=-[2\sin\omega(k)]^{-1}
\partial_k\Tr[\tilde U(k)]$, whereas the second gives
$\hat n(k)=i[2\sin\omega(k)]^{-1}
\Tr[\vec\sigma\tilde U(k)]$. Substituting both relations into
\Eref{eq:T_definition} and using
$\sin^2\omega(k)=[4-\bqty{\Tr[\tilde U(k)]}^2]/4$ yields
\begin{equation}\protect\label{eq:vmcT_universal_app1}
\vmcT=i\int_{-\pi}^{\pi}\frac{dk}{2\pi}
\frac{\partial_k\Tr[\tilde U(k)]
\Tr[\vec\sigma\tilde U(k)]}
{\bqty{\Tr[\tilde U(k)]}^2-4}.
\end{equation}

% We now specialize this expression to $\tilde U(k)=S(k)C$, with
% $S(k)=e^{-ik\sigma_z}$ and $C\in\mathrm{SU}(2)$, and evaluate the
% remaining integral in terms of the matrix elements of $C$.
% 
We now specialize this expression to a single-step walk with an arbitrary
coin operator $C\in\mathrm{U}(2)$. Let $\widehat C\in\mathrm{SU}(2)$ denote
its determinant-one representative. Under the convention above,
$\tilde U(k)=S(k)\widehat C$, with $S(k)=e^{-ik\sigma_z}$.
We evaluate the remaining integral using the matrix elements of
$\widehat C$.
% 
% We parameterize the single coin operator $C\in\mathrm{SU}(2)$ by its
% upper-row matrix elements $C_{00}$ and $C_{01}$, with lower-row elements
% fixed by unitarity as $C_{11}=C_{00}^*$ and
% $C_{10}=-C_{01}^*$. 
% 
We parameterize $\widehat C$ by its upper-row matrix
elements $\widehat C_{00}$ and $\widehat C_{01}$, with lower-row elements
fixed by unitarity as $\widehat C_{11}=\widehat C_{00}^*$ and 
$\widehat C_{10}=-\widehat C_{01}^*$.
For the shift operator $S(k)=e^{-i\sigma_z k}$, the matrix multiplication
$\tilde U(k)=S(k)\widehat C$ yields:
\begin{equation}\label{eq:trace_identities_app1}
\begin{split}
\Tr[\tilde U(k)] &= e^{-ik} \widehat C_{00} + e^{ik} \widehat C_{00}^*, \\
\pdv{\Tr[\tilde U(k)]}{k} &= -i \bqty{ e^{-ik} \widehat C_{00} - e^{ik} \widehat C_{00}^* }, \\
\Tr[\vec{\sigma} \tilde U(k)] &= \mqty(
e^{-ik} \widehat C_{01} - e^{ik} \widehat C_{01}^* \\
i \pqty{ e^{-ik} \widehat C_{01} + e^{ik} \widehat C_{01}^* } \\
e^{-ik} \widehat C_{00} - e^{ik} \widehat C_{00}^*
).
\end{split}
\end{equation}
Substituting \Eref{eq:trace_identities_app1} into
\Eref{eq:vmcT_universal_app1}, and writing the diagonal coin-operator
element in polar form, we use
% \codexnote{
% $\widehat C_{00}=\abs{\widehat C_{00}}e^{i\alpha}$
% }{
$\widehat C_{00}=\abs{\widehat C_{00}}e^{i\beta}$.
% }.
We then factor $\abs{\widehat C_{00}}$ inside the vector using
% \codexnote{
% $\abs{\widehat C_{00}}=\widehat C_{00}^*e^{i\alpha}=\widehat C_{00}e^{-i\alpha}$
% }{
$\abs{\widehat C_{00}}=\widehat C_{00}^*e^{i\beta}=\widehat C_{00}e^{-i\beta}$, 
which aligns the phases across all components. Performing the variable
transformation
% \codexnote{
% $u=k-\alpha$
% }{
$u=k-\beta$,
the $2\pi$-periodicity of the integrand allows us to replace the shifted
integration bounds directly with the symmetric interval $[-\pi,\pi]$:
% \codexreason{The symbol $\alpha$ already denotes the global
% $\mathrm{U}(1)$ phase; $\beta$ avoids conflating it with the polar phase of
% $\widehat C_{00}$.}
\begin{equation}\label{eq:u_integral_app1}
\begin{split}
\vmcT &= \int_{-\pi}^{\pi} \frac{du}{2\pi} \frac{-2i \abs{\widehat C_{00}} \sin u}{\bqty{ 2 \abs{\widehat C_{00}} \cos u }^2 - 4} \frac{1}{\abs{\widehat C_{00} }} \mqty(
e^{-iu} \widehat C_{01} \widehat C_{00}^* - e^{iu} \widehat C_{01}^* \widehat C_{00} \\
i \pqty{ e^{-iu} \widehat C_{01} \widehat C_{00}^* + e^{iu} \widehat C_{01}^* \widehat C_{00} } \\
-2i \abs{\widehat C_{00}}^2 \sin u
) \\
&= \int_{-\pi}^{\pi} \frac{du}{2\pi} \frac{-2i \abs{\widehat C_{00}} \sin u}{\bqty{ 2 \abs{\widehat C_{00}} \cos u }^2 - 4} \frac{1}{\abs{\widehat C_{00} }} \mqty(
-2i \sin u \, \Re(\widehat C_{01} \widehat C_{00}^*) + 2i \cos u \, \Im(\widehat C_{01} \widehat C_{00}^*) \\
2i \sin u \, \Im(\widehat C_{01} \widehat C_{00}^*) + 2i \cos u \, \Re(\widehat C_{01} \widehat C_{00}^*) \\
-2i \abs{\widehat C_{00}}^2 \sin u
).
\end{split}
\end{equation}
The terms in the vector proportional to $\cos u$ vanish upon integration:
multiplication by the common scalar prefactor, which is odd in $u$, makes
their full integrands odd over the symmetric interval $[-\pi,\pi]$.
The remaining terms, proportional to $\sin u$, yield even integrands.
Factoring out their common $u$ dependence gives
\begin{equation}\label{eq:canonical_integral_app1}
    \vmcT = 
    \pqty{ 
        \int_{-\pi}^{\pi} 
        \frac{du}{2\pi} 
        \frac{\sin^2 u}{1 - \abs{\widehat C_{00}}^2 \cos^2 u}
    } \mqty(
        \Re(\widehat C_{01} \widehat C_{00}^*) \\
        -\Im(\widehat C_{01} \widehat C_{00}^*) \\
        \abs{\widehat C_{00}}^2
    ).
\end{equation}

To obtain the final analytical expression, we evaluate the scalar integral
appearing in parenthesis in \Eref{eq:canonical_integral_app1}. Defining this
integral as $\mathcal{I}_0$, direct integration followed by the unitarity
condition $\abs{\widehat C_{00}}^2+\abs{\widehat C_{01}}^2=1.$ gives
\begin{equation}\label{eq:scalar_integral_evaluated}
\mathcal{I}_0 = \frac{1 - \sqrt{1 - \abs{\widehat C_{00}}^2}}{\abs{\widehat C_{00}}^2} = \frac{1 - \abs{\widehat C_{01} }}{\abs{\widehat C_{00}}^2}.
\end{equation}
Multiplying the numerator and denominator of \Eref{eq:scalar_integral_evaluated} by the conjugate factor $\pqty{1 + \abs{\widehat C_{01} }}$ simplifies to
\begin{equation}\label{eq:scalar_integral_simplified}
\mathcal{I}_0 = \frac{1 - \abs{\widehat C_{01}}^2}{\abs{\widehat C_{00}}^2 \pqty{1 + \abs{\widehat C_{01} } }} = \frac{\abs{\widehat C_{00}}^2}{\abs{\widehat C_{00}}^2 \pqty{1 + \abs{\widehat C_{01}} }} = \frac{1}{1 + \abs{\widehat C_{01}} }.
\end{equation}
Substituting \Eref{eq:scalar_integral_simplified} back into \Eref{eq:canonical_integral_app1} provides the exact closed-form expression for the asymptotic transport vector in \Eref{eq:transport_vector_integrated} of the main text:
\begin{equation}\label{eq:transport_vector_integrated_SM}
  \vmcT = \frac{1}{1+\abs{\widehat C_{01} }}
  \begin{pmatrix}
    \Re(\widehat C_{01}\, \widehat C_{00}^*) \\
    -\Im(\widehat C_{01}\, \widehat C_{00}^*) \\
    \abs{\widehat C_{00}}^2
  \end{pmatrix}.
\end{equation}
Since this expression is invariant under a global phase of the coin
operator, it can equivalently be written in terms of the matrix elements
of the original $C\in\mathrm{U}(2)$, yielding 
Eq.~\eqref{eq:transport_vector_integrated} of the main text.  
Thus, without loss of generality, throughout the remainder of the 
Supplemental Material we denote by $C$ (and $C_j$ when several coin 
operators are involved) the corresponding $\mathrm{SU}(2)$ representative.
%*}}}
\section{General residue-theorem method for the transport vector} % {{{*
\label{sm:residue_method}
In Sec.~\ref{sm:transport_vector_integral}, we derived the general expression \Eref{eq:vmcT_universal_app1} for the transport vector solely in terms of the unitary step operator $\tilde U(k)$. There, we evaluated this expression for a DTQW with $\tilde U(k)=S(k)C$ by performing the momentum integral directly. Here, we present a residue-theorem method for evaluating the transport vector that can be applied to more general step operators $\tilde U(k)$. Although we illustrate the procedure explicitly for the single-step evolution $\tilde U(k)=S(k)C$, the method also applies to more general compositions, such as $\tilde U(k)=S(k)C_2S(k)C_1$, and, more generally, to evolutions involving multiple shift and coin operators. For sufficiently complicated compositions, an analytical evaluation of the resulting contour integral may become impractical, as it requires determining the zeros of high-degree polynomials. Nevertheless, the same procedure can still be used to evaluate the transport vector numerically.
For a single coin operator $C$ we find
% \begin{equation}
%     C = \begin{pmatrix}
%         C_{00} & C_{01}\\
%         C_{10} & C_{11}
%     \end{pmatrix}\, ,
% \end{equation}
\begin{equation}
    \tilde U(k) = S(k)\cdot C = \begin{pmatrix}
        e^{-ik}C_{00} &e^{-ik}C_{01} \\
        e^{ik}C_{10} & e^{ik}C_{11}
    \end{pmatrix}\, ,
\end{equation}
where $C_{ij}$ are the matrix elements of $C$, implying that we find the transport vector in \Eref{eq:vmcT_universal_app1} through the integral
\begin{equation}
    \vec{\mcT} = \int_{-\pi}^\pi \frac{dk}{2 \pi} \frac{1}{\left(C_{00}+C_{11} e^{2 i k}\right)^2-4 e^{2 i k}}
    \begin{pmatrix}
    \left(C_{00}-C_{11} e^{2 i k}\right) \left(C_{01}+C_{10} e^{2 i k}\right)\\
    i \left(C_{00}-C_{11} e^{2 i k}\right) \left(C_{01}-C_{10} e^{2 i k}\right)\\
    \left(C_{00}-C_{11} e^{2 i k}\right)^2
    \end{pmatrix} \, ,
\end{equation}
where for simplicity we will now show how to solve these integrals only on the $x$-component, but all other components can be derived in the same manner.

Our plan of action is to use the residue theorem in the complex plane, for which we want to use the unit circle as path, which we parametrize by
$\gamma(t) = e^{it}$, which implies that we can write the integral as
\begin{equation}
    \mcT_x = \int_{-\pi}^\pi \frac{dk}{2 \pi}
\frac{\left(C_{00}-C_{11} e^{2 i k}\right) \left(C_{01}+C_{10} e^{2 i k}\right)}{\left(C_{00}+C_{11} e^{2 i k}\right)^2-4 e^{2 i k}} = \oint_{\gamma} \frac{dz}{2 \pi} -\frac{i \left(C_{00}-C_{11} z ^2\right) \left(C_{01}+C_{10} z ^2\right)}{z  \left(\left(C_{00}+C_{11} z ^2\right)^2-4 z ^2\right)} \, ,
\label{eq: Tx: k integral to line integral}
\end{equation}
where the closed line integral is for $t\in [-\pi, \pi]$, i.e., anti-clockwise. The integrand of the line integral has poles at
\begin{align}
    z_0 = 0 &&
    z_{1,2} = \pm \frac{1 + \sqrt{1-C_{00} C_{11}}}{C_{11}} &&
    z_{3,4} = \pm \frac{1 - \sqrt{1-C_{00} C_{11}}}{C_{11}}
\end{align}
where trivially $z_0$ is inside the unit circle and to check the others we first use that $C_{11} = C_{00}^*$ and $0\leq|C_{00}|\leq 1$, thus we need to check
\begin{equation}
   |z_{1,2,3,4}| \leq 1\; \Leftrightarrow \; 1\pm\sqrt{1-|C_{00}|^2}\leq |C_{00}| \leq 1 \, ,
\end{equation}
which is never possible for $z_{1,2}$ as $1+\sqrt{1-|C_{00}|^2}\geq1$ and the inequality holds for $z_{3,4}$, as it is equivalent to $|C_{00}|-1\leq 0$. Therefore the relevant poles are $z_0$, $z_3$ and $z_4$ and as all poles are first order finding the residues is straightforward, thus $\mcT_x$ is given by
\begin{equation}
    \mcT_x = \oint_{\gamma} \frac{dz}{2\pi} f(z) =  i \left(\text{Res}_f(z_0)+\text{Res}_f(z_3)+\text{Res}_f(z_4) \right) = \frac{\left(\sqrt{1-C_{00} C_{11}}-1\right) (C_{00} C_{10}-C_{01} C_{11})}{2 C_{00} C_{11}} \, ,
\end{equation}
where $f(z) = -\frac{i \left(C_{00}-C_{11} z ^2\right) \left(C_{01}+C_{10} z ^2\right)}{z  \left(\left(C_{00}+C_{11} z ^2\right)^2-4 z ^2\right)}$. 
Using $1-|C_{00}|^2=|C_{01}|^2$ and $C_{10}=-C_{01}^*$ to simplify $f(z)$ gives
\begin{equation}
      \mcT_x = \frac{\Re(C_{00} C^*_{01})}{1+|C_{01}|} \, .
\end{equation}
The remaining components follow in the same manner, with the same poles
and only a different numerator, yielding
Eq.~\eqref{eq:transport_vector_integrated}. As shown in
Sec.~\ref{sm:transport_vector_integral}, this expression is invariant under
a global phase of the coin operator and therefore also holds for the original 
coin operator in $\mathrm{U}(2)$.

We further want to outline that these steps of first writing down the integral in terms of $e^{ik}$ and then going to the complex plane and integrating over the line of the unit circle, is a general recipe that works also for more complex setups, for example for combining two coin operators or unitary steps. Therefore we want to outline each step in the following. First, we obtain the integral directly from \Eref{eq:vmcT_universal_app1}. We then make the substitution $e^{ik}\rightarrow z$ and use $dk=dz/(iz)$, which transforms the momentum integral into a contour integral along the unit circle $\gamma$, as in \Eref{eq: Tx: k integral to line integral}. All that is left to do now is find the poles of the resulting integrand and check which ones lie inside the unit circle and calculate the residues of those, which is easy by factorizing the denominators. Therefore the only complicated step is finding the zeros of a higher-order polynomial, which is numerically easy.

% *}}}
\section{Collinearity of the two-step DTQW transport vector} % {{{*
\label{sm:colinearity}

We now consider a two-step DTQW governed by the composite evolution
operator $U_3(k)=S(k)C_2S(k)C_1$.
% \codexinsert{
% Consistently with the convention above, $C_1$, $C_2$, and $U_3(k)$ denote
% their determinant-one $\mathrm{SU}(2)$ representatives throughout this
% section.
% }
Our objective is to demonstrate that its asymptotic transport vector
$\vmcT_3$ is collinear with the single-step transport vector
$\vmcT$ derived in Appendix~\ref{sm:transport_vector_integral}. Invoking the
formula established in \Eref{eq:vmcT_universal_app1}, the two-step transport
vector reads:
\begin{equation}\label{eq:vmcT_universal_app2}
\vmcT_3 = i \int_{-\pi}^{\pi} \frac{dk}{2\pi} \frac{1}{\bqty{\Tr[U_3(k)]}^2 - 4} \pdv{\Tr[U_3(k)]}{k} \Tr[\vec{\sigma} U_3(k)].
\end{equation}

To distinguish between the two coin operators while preserving consistent matrix-element notation, we label them by $j\in\{1,2\}$ and denote the upper-row elements of $C_j$ by $C_j^{00}$ and $C_j^{01}$. In the two-step sequence, the shift operator $S(k) = e^{-i\sigma_z k}$ acts twice, generating phase evolutions governed by $2k$. Defining the constant, momentum-independent scalar $C_0 = C_2^{01} \pqty{C_1^{01}}^* + \pqty{C_2^{01}}^* C_1^{01} = 2\Re\bqty{ C_2^{01} \pqty{C_1^{01}}^* }$, matrix multiplication yields the following:
\begin{equation}\label{eq:trace_identities_app2}
\begin{split}
\Tr[U_3(k)] &= e^{-2ik} C_2^{00} C_1^{00} + e^{2ik} \pqty{C_2^{00}}^* \pqty{C_1^{00}}^* - C_0, \\
\pdv{\Tr[U_3(k)]}{k} &= -2i \bqty{ e^{-2ik} C_2^{00} C_1^{00} - e^{2ik} \pqty{C_2^{00}}^* \pqty{C_1^{00}}^* }, \\
\Tr[\vec{\sigma} U_3(k)] &= \mqty(
e^{-2ik} C_2^{00} C_1^{01} + C_2^{01} \pqty{C_1^{00}}^* - \pqty{C_2^{01}}^* C_1^{00} - e^{2ik} \pqty{C_2^{00}}^* \pqty{C_1^{01}}^* \\
-i \bqty{ e^{-2ik} C_2^{00} C_1^{01} + e^{2ik} \pqty{C_2^{00}}^* \pqty{C_1^{01}}^* + C_2^{01} \pqty{C_1^{00}}^* + \pqty{C_2^{01}}^* C_1^{00} } \\
e^{-2ik} C_2^{00} C_1^{00} - e^{2ik} \pqty{C_2^{00}}^* \pqty{C_1^{00}}^* - C_2^{01} \pqty{C_1^{01}}^* + \pqty{C_2^{01}}^* C_1^{01}
).
\end{split}
\end{equation}

Substituting \Eref{eq:trace_identities_app2} into \Eref{eq:vmcT_universal_app2} separates the integral naturally into two distinct contributions: a momentum-dependent vector, and a constant term containing the momentum-independent vector components proportional to $C_2^{01}$:
\begin{equation}\label{eq:split_integral_app2}
\begin{split}
\vmcT_3 &= \int_{-\pi}^{\pi} \frac{dk}{2\pi} \frac{2 \bqty{ e^{-2ik} C_2^{00} C_1^{00} - e^{2ik} \pqty{C_2^{00}}^* \pqty{C_1^{00}}^* }}{\bqty{ \Tr[U_3(k)] }^2 - 4} \mqty(
e^{-2ik} C_2^{00} C_1^{01} - e^{2ik} \pqty{C_2^{00}}^* \pqty{C_1^{01}}^* \\
-i \bqty{ e^{-2ik} C_2^{00} C_1^{01} + e^{2ik} \pqty{C_2^{00}}^* \pqty{C_1^{01}}^* } \\
e^{-2ik} C_2^{00} C_1^{00} - e^{2ik} \pqty{C_2^{00}}^* \pqty{C_1^{00}}^*
) \\
&\quad + \int_{-\pi}^{\pi} \frac{dk}{2\pi} \frac{2 \bqty{ e^{-2ik} C_2^{00} C_1^{00} - e^{2ik} \pqty{C_2^{00}}^* \pqty{C_1^{00}}^* }}{\bqty{ \Tr[U_3(k)] }^2 - 4} \mqty(
C_2^{01} \pqty{C_1^{00}}^* - \pqty{C_2^{01}}^* C_1^{00} \\
-i \bqty{ C_2^{01} \pqty{C_1^{00}}^* + \pqty{C_2^{01}}^* C_1^{00} } \\
- C_2^{01} \pqty{C_1^{01}}^* + \pqty{C_2^{01}}^* C_1^{01}
).
\end{split}
\end{equation}
Notice that the scalar numerator of the second integral is directly proportional to the exact differential of the trace [see \Eref{eq:trace_identities_app2}]: $2\bqty{e^{-2ik} C_2^{00} C_1^{00} - e^{2ik} \pqty{C_2^{00}}^* \pqty{C_1^{00}}^*} = i \, \partial{\Tr[U_3(k)]}/\partial{k}$. 
Since $U_3(k)\in\mathrm{SU}(2)$, its trace is real and can be written as
$\Tr[U_3(k)]=2\cos\omega(k)$. Thus, away from the isolated points where
$\Tr[U_3(k)]=\pm2$, the scalar integrand has the real antiderivative
\begin{equation}
\frac{i}{\Tr[U_3(k)]^2-4}
\pdv{\Tr[U_3(k)]}{k}
=
-i\pdv{}{k}
\left[
\frac{1}{2}
\operatorname{arctanh}
\left(
\frac{\Tr[U_3(k)]}{2}
\right)
\right].
\end{equation}
Since $U_3(-\pi)=U_3(\pi)$, see \Eref{eq:trace_identities_app2}, the integral 
therefore vanishes.

For the surviving term of the integral in \Eref{eq:split_integral_app2}, we express both diagonal coin-operator elements in polar form as $C_1^{00} = \abs{C_1^{00}} e^{i\alpha_1}$ and $C_2^{00} = \abs{C_2^{00}} e^{i\alpha_2}$. 
Factoring $\abs{C_2^{00}}$ and $\abs{C_1^{00}}$ out and applying 
the polar identities $\abs{C_1^{00}} = \pqty{C_1^{00}}^* e^{i\alpha_1} = C_1^{00} e^{-i\alpha_1}$ aligns the off-diagonal phases across all vector components, exactly as in Appendix~\ref{sm:transport_vector_integral}. We then apply the variable transformation $u = 2k - \alpha_2 - \alpha_1$. As $k$ spans $[-\pi, \pi]$, the angle $u$ covers exactly two periods of $2\pi$ of the integrand. This factor of $2$ arising from integrating over two periods cancels the factor coming from the change of variables, $dk = du/2$, yielding:
\begin{equation}\label{eq:u_integral_app2}
\begin{split}
\vmcT_3 &= -4i \abs{C_2^{00}}^2 \int_{-\pi}^{\pi} \frac{du}{2\pi} \frac{\sin u}{\bqty{ 2 \abs{C_2^{00}} \abs{C_1^{00}} \cos u - C_0 }^2 - 4} \mqty(
e^{-iu} C_1^{01} \pqty{C_1^{00}}^* - e^{iu} \pqty{C_1^{01}}^* C_1^{00} \\
i \bqty{ e^{-iu} C_1^{01} \pqty{C_1^{00}}^* + e^{iu} \pqty{C_1^{01}}^* C_1^{00} } \\
-2i \abs{C_1^{00}}^2 \sin u
) \\
&= -4i \abs{C_2^{00}}^2 \int_{-\pi}^{\pi} \frac{du}{2\pi} \frac{\sin u}{\bqty{ 2 \abs{C_2^{00}} \abs{C_1^{00}} \cos u - C_0 }^2 - 4} \mqty(
-2i \sin u \, \Re(C_1^{01} \pqty{C_1^{00}}^*) + 2i \cos u \, \Im(C_1^{01} \pqty{C_1^{00}}^*) \\
2i \cos u \, \Re(C_1^{01} \pqty{C_1^{00}}^*) + 2i \sin u \, \Im(C_1^{01} \pqty{C_1^{00}}^*) \\
-2i \abs{C_1^{00}}^2 \sin u
).
\end{split}
\end{equation}
Because the cosine is an even function, the transformed denominator is strictly even under parity ($u \to -u$). Multiplying by the odd numerator prefactor $\sin u$ renders the overall scalar function strictly odd. Invoking parity symmetry, as we did in Appendix~\ref{sm:transport_vector_integral}, the even $\cos u$ vector terms vanish identically upon integration over $[-\pi, \pi]$. Retaining only the odd $\sin u$ vector terms and factoring out $-2i \sin u$ produces an overall prefactor of $-8 \abs{C_2^{00}}^2$ and isolates the integral form for the two-step walk:
\begin{equation}\label{eq:canonical_integral_app2}
\vmcT_3 =  \bqty{ -8 \abs{C_2^{00}}^2 \int_{-\pi}^{\pi} \frac{du}{2\pi} \frac{\sin^2 u}{\bqty{ 2 \abs{C_2^{00}} \abs{C_1^{00}} \cos u - C_0 }^2 - 4} } \mqty(
\Re(C_1^{01} \pqty{C_1^{00}}^*) \\
-\Im(C_1^{01} \pqty{C_1^{00}}^*) \\
\abs{C_1^{00}}^2
).
\end{equation}

A direct comparison between \Eref{eq:canonical_integral_app2} and the single-step integral form established in \Eref{eq:canonical_integral_app1} of Appendix~\ref{sm:transport_vector_integral} reveals an exact correspondence. In both DTQWs, all coin-operator-dependent vector components factor completely outside the integral before scalar evaluation, converging onto the identical direction vector $\vec{u}_{C_1}$ fixed exclusively by the matrix elements of coin operator $C_1$:
\begin{equation}\label{eq:universal_vector_app2}
\vec{u}_{C_1} = \mqty(
\Re(C_1^{01} \pqty{C_1^{00}}^*) \\
-\Im(C_1^{01} \pqty{C_1^{00}}^*) \\
\abs{C_1^{00}}^2
).
\end{equation}
Consequently, the two-step transport vector $\vmcT_3$
is a nonnegative scalar multiple of
the single-step transport vector derived in
Appendix~\ref{sm:transport_vector_integral}, with proportionality coefficient
$\gamma$ given by the ratio of their respective scalar integrals:
\begin{equation}\label{eq:gamma_ratio_app2}
\gamma = 2 \abs{C_2^{00}}^2 \frac{\int_{-\pi}^{\pi} \frac{du}{2\pi} \frac{\sin^2 u}{\bqty{ 2 \abs{C_2^{00}} \abs{C_1^{00}} \cos u - C_0 }^2 - 4}}{\int_{-\pi}^{\pi} \frac{du}{2\pi} \frac{\sin^2 u}{\bqty{ 2 \abs{C_1^{00}} \cos u }^2 - 4}}.
\end{equation}
For the $\mathrm{SU}(2)$ part of either walk,
$\bqty{\Tr U(k)}^2-4\leq0$. Thus, both scalar integrals in
Eq.~\eqref{eq:gamma_ratio_app2} have the same nonpositive sign.
Since $2\abs{C_2^{00}}^2\geq0$, it follows that $\gamma\geq0$.
Degenerate endpoint cases follow by continuity.

% *}}}
\section{Reduced coin state dynamics}  % Reduced coin state dynamics {{{*
\label{sec:reduced_coin}

The primary objective of this section is to analyze the reduced dynamics of the internal coin space and establish an exact mathematical connection between its asymptotic stationary state and the transport vector $\vmcT$.

We begin from the momentum-space state at an arbitrary discrete time step
$t$, derived in \Eref{eq:momentum_state_evolution}. Dropping the
$k$-independent global phase $e^{-i\alpha t}$, which has no effect on the
reduced density matrix, the state vector reads
% \codexinsert{
% Here, $\hat C=e^{i\alpha}C$ denotes the determinant-one coin representative,
% so that $\tilde U(k)=S(k)\hat C$ under the convention above.
% }
\begin{equation}
    \ket{\tilde\psi_t(k)} = c_+(k)e^{-i\omega(k)t} \ket{\phi_+(k)} + c_-(k)e^{i\omega(k)t} \ket{\phi_-(k)},
    \label{eq:momentum_state}
\end{equation}
where $\pm\omega(k)$ and $\ket{\phi_\pm(k)}$ denote the quasienergies and
eigenstates of
% \codexnote{
$\tilde U(k)=S(k)C$,
% }{
% $\tilde U(k)=S(k)\hat C$
respectively, and
$c_\pm(k)=\braket{\phi_\pm(k)}{\tilde\psi_0}$ are the amplitudes of the
initial state onto the eigenstates of $\tilde U(k)$.
To obtain the reduced density matrix of the coin, $\rho_c(t)$, we trace out the lattice space. In the momentum representation, this partial trace corresponds to integrating over $k \in [-\pi, \pi)$:
\begin{equation}
    \rho_c(t) = \int_{-\pi}^{\pi} \frac{dk}{2\pi} \dyad{\tilde\psi_t(k)}.
    \label{eq:reduced_trace}
\end{equation}
Expanding the outer product using Eq.~\eqref{eq:momentum_state} naturally separates the reduced density matrix into two distinct components: a time-independent stationary contribution and a transient interference term,
\begin{equation}
    \rho_c(t) = \rho_{\text{stat}} + \rho_{\text{int}}(t),
    \label{eq:reduced_coin_split}
\end{equation}
where explicitly:
\begin{align}
    \rho_{\text{stat}} &= \int_{-\pi}^{\pi} \frac{dk}{2\pi} \pqty{ \abs{c_+(k)}^2 \dyad{\phi_+(k)} + \abs{c_-(k)}^2 \dyad{\phi_-(k)} }, \label{eq:stationary_state} \\
    \rho_{\text{int}}(t) &= \int_{-\pi}^{\pi} \frac{dk}{2\pi} \pqty{ c_+(k)c^*_-(k) e^{-2i\omega(k)t} \dyad{\phi_+(k)}{\phi_-(k)} + \text{h.c.} }. \label{eq:transient_state}
\end{align}
While $\rho_{\text{int}}(t)$ dictates the short-time oscillatory behavior, it vanishes in the asymptotic limit $t \to \infty$. Each matrix element of $\rho_{\mathrm{int}}(t)$ is an integral of the form $\int dkf(k)e^{\pm 2i\omega(k)t}$, where $f(k)$ is time independent. As $t$ increases, a small change in $k$ produces a large change in the phase $2\omega(k)t$, while the prefactor $f(k)$ changes very little over the same range of $k$. Thus, over short momentum intervals, the integral effectively adds almost equal amplitudes with different complex phases, which cancel each other in the long-time limit.

To evaluate the steady state explicitly, we express the spectral projection operators and the projection coefficients in their Bloch form on the Bloch sphere:
\begin{align}
    \dyad{\phi_\pm(k)} &= \frac{1}{2} \pqty{ \mathbb{I} \pm \hat{n}(k) \cdot \vec{\sigma} }, \label{eq:bloch_projectors} \\
    \abs{c_\pm(k)}^2 &= \frac{1}{2} \pqty{ 1 \pm \vec{r}_0 \cdot \hat{n}(k) }, \label{eq:bloch_populations}
\end{align}
where $\hat{n}(k)$ is the real unit spectral vector defining the rotation axis at momentum $k$, $\vec{\sigma}$ is the vector of Pauli matrices, and $\vec{r}_0$ is the initial Bloch vector of the coin. Substituting Eqs.~\eqref{eq:bloch_projectors} and \eqref{eq:bloch_populations} into Eq.~\eqref{eq:stationary_state}, the cross terms cancel symmetrically.
This yields an integral form for the steady state expressed entirely in terms
of the components of $\hat{n}(k)$:
\begin{equation}
\begin{split}
    \rho_{\text{stat}}
    &= \frac{1}{2}\mathbb{I}
    + \frac{1}{2}
    \sum_{i,j=x,y,z}
    r_{0,j}
    \pqty{
    \int_{-\pi}^{\pi}\frac{dk}{2\pi}
    n_i(k)n_j(k)
    }
    \sigma_i \\
    &= \frac{1}{2}\mathbb{I}
    + \frac{1}{2}
    \sum_{i,j=x,y,z}
    \mathcal{M}_{ij} r_{0,j}\sigma_i .
\end{split}
\label{eq:stationary_state_expanded}
\end{equation}
Here $\mathcal{M}_{ij} = \int_{-\pi}^{\pi}\frac{dk}{2\pi}n_i(k)n_j(k)$ are the matrix elements of the real symmetric matrix $\mathcal{M}=\int_{-\pi}^{\pi}\frac{dk}{2\pi}\hat{n}(k)\hat{n}^{T}(k)$. Thus, the stationary Bloch vector of the coin is obtained by applying $\mathcal{M}$ to the initial Bloch vector $\vec{r}_0$.

To establish the connection of $\mathcal M$ with the transport vector
$\vmcT$, we consider the derivative of
$\cos\omega(k)=\frac{1}{2}\Tr\bqty{\tilde U(k)}$ with respect to
quasi-momentum:
\begin{equation}
    -\sin\omega(k) \partial_k \omega(k)
    =
    \frac{1}{2}\Tr\bqty{\partial_k\pqty{e^{-ik\sigma_z}C}}
    =
    -\frac{i}{2}\Tr\bqty{\sigma_z \tilde{U}(k)}.
    \label{eq:trace_derivative}
\end{equation}
Expanding the evolution operator as
$\tilde U(k)=\cos\omega(k)\mathbb I
-i\sin\omega(k)\hat n(k)\cdot\vec\sigma$, we obtain
\begin{equation}
    -\frac{i}{2}
    \Tr\bqty{
    \sigma_z
    \pqty{
    \cos\omega(k)\mathbb{I}
    -i\sin\omega(k)\hat{n}(k)\cdot\vec{\sigma}
    }
    }
    =
    -\sin\omega(k)n_z(k).
    \label{eq:trace_evaluation}
\end{equation}
Equating Eqs.~\eqref{eq:trace_derivative} and
\eqref{eq:trace_evaluation} gives
\begin{equation}
    \partial_k\omega(k)=n_z(k).
    \label{eq:group_velocity_identity}
\end{equation}
Substituting this identity into the definition of the transport vector in
\Eref{eq:T_definition}, its $i$-th component becomes
\begin{equation}
    \vmcT_i
    =
    \int_{-\pi}^{\pi}\frac{dk}{2\pi}
    n_z(k)n_i(k)
    =
    \mathcal{M}_{iz}.
    \label{eq:transport_tensor_connection}
\end{equation}
Equation~\eqref{eq:stationary_state_expanded} shows that $\mathcal{M}$
maps the initial Bloch vector to the stationary Bloch vector of the coin,
whereas Eq.~\eqref{eq:transport_tensor_connection} shows that the transport
vector is the $z$ column of the same matrix.  Since $\mathcal{M}$ is symmetric,
the $z$ component of the stationary Bloch vector satisfies
\begin{equation*}
    (\vec{r}_{\mathrm{stat}})_z
    =
    \sum_{j=x,y,z}\mathcal{M}_{zj}r_{0,j}
    =
    \sum_{j=x,y,z}\mathcal{M}_{jz}r_{0,j}
    =
    \vmcT\cdot\vec{r}_0
    =
    \overline{v}.
\end{equation*}
Therefore, the coin's steady state directly encodes the walker's asymptotic
directed spatial transport through the $z$ component of its stationary Bloch
vector.

% *}}}
\section{Derivation of critical rotation angle} % {{{*
\label{SM:critical_rot_angle}
% Intro   % {{{*
Consider a strategy given by a fixed initial coin state, whose Bloch vector $\vec{r}_0(\theta, \phi)$ is parametrized by its polar and azimuthal angles $\theta$ and $\phi$, alongside a coin operator $C=\exp\left[-i (\chi/2) \hat n \cdot \vec \sigma\right] =C(\hat n, \chi)$ with a fixed rotation axis $\hat{n}(\vartheta, \varphi)$ parametrized by its polar and azimuthal angles $\vartheta$ and $\varphi$. For a generic initial state, the rotation angles $\chi$ for which the position expectation value vanishes may depend on the time $t$. However, when the initial state lies on the equator of the Bloch sphere, there exist specific rotation angles $\chi$ for which the expectation value of the position vanishes for all times $t\geq 0$. We refer to the unique non-trivial (non-zero) rotation angle that satisfies this condition for a given configuration as the critical rotation angle, $\chi_\mathrm{c}$. In this supplemental material, we derive the mathematical expression for $\chi_\mathrm{c}$ for any fixed configuration $\{\vartheta, \varphi, \theta=\pi/2, \phi\}$ through two different approaches. In Sec.~\ref{chp:critical_angle_transport_vector}, we find the rotation angles for which the expectation value of the position vanishes in the asymptotic limit using the transport vector formalism. In Sec.~\ref{chp:critical_angle_konno}, we find the rotation angles for which the expectation value of the position vanishes for all times $t \geq 0$ using a theorem that dictates the necessary and sufficient conditions required for the probability distribution to remain symmetric at any time $t$. Both approaches independently yield, as expected, the same mathematical expression for the critical angle.
% *}}}
\subsection{Critical rotation angle by setting the asymptotic velocity to zero}\label{chp:critical_angle_transport_vector}  % {{{*
We now derive the critical rotation angle $\chi_\mathrm{c}$ by setting the asymptotic velocity to zero, $\overline{v} = \vmcT \cdot \vec{r}_0 = 0$, and considering an initial coin state whose Bloch vector is constrained to the $xy$-plane of the Bloch sphere, $\vec{r}_0(\theta = \pi/2, \phi)$, for the walker to have an equal probability of stepping left or right initially.

For a coin operator $C=\exp[-i(\chi/2)\hat n\cdot\vec\sigma]$ with a fixed rotation axis $\hat n(\vartheta,\varphi)$, the transport vector [Eq.~\eqref{eq:transport_vector_integrated_SM}] can be rewritten in terms of its matrix elements, $C_{00}=\cos(\chi/2)-i\cos\vartheta\sin(\chi/2)$ and $C_{01}=-ie^{-i\varphi}\sin\vartheta\sin(\chi/2)$, as follows:
\begin{equation}\label{eq:transp_vec}
    \vmcT = \frac{1}{1+\sin\vartheta \sin(\chi/2)}
    \begin{pmatrix}
        \sin\vartheta \sin(\chi/2) \big[\cos\vartheta \sin(\chi/2)\cos\varphi - \cos(\chi/2)\sin\varphi\big] \\[0.25em]
        \sin\vartheta \sin(\chi/2) \big[\cos(\chi/2)\cos\varphi + \cos\vartheta \sin(\chi/2)\sin\varphi\big] \\[0.25em]
        1 - \sin^2\vartheta \sin^2(\chi/2)
    \end{pmatrix}.
\end{equation}
Using this expression and applying trigonometric angle addition identities, the asymptotic velocity for an initial coin state on the equator of the Bloch sphere ($\theta = \pi/2$) is given by
\begin{equation}\label{eq:asymp_vel_xyplane}
    \overline{v} = \frac{\sin\vartheta \sin(\chi/2)}{1+\sin\vartheta \sin(\chi/2)} \big[\sin(\phi - \varphi) \cos(\chi/2) + \cos(\phi - \varphi) \cos\vartheta \sin(\chi/2)\big].
\end{equation}
Because the factor $1/\big[1+\sin\vartheta \sin(\chi/2)\big]$ is strictly non-zero, the velocity vanishes only when the remaining factors equal zero.  Then, we can obtain the critical rotation angle $\chi_\mathrm{c}$ by finding the unique nontrivial solution ($\chi \not= 0$) to
\begin{equation}\label{eq:critical_angle}
    \sin\vartheta \sin(\chi/2)\big[\sin(\phi - \varphi) \cos(\chi/2) + \cos(\phi - \varphi) \cos\vartheta \sin(\chi/2)\big] = 0.
\end{equation}
To guarantee a unique non-zero solution to the previous equation, two conditions must be met: (1) the coin operator's rotation axis must not align with the $z$-axis ($\sin\vartheta \neq 0$), and (2) its projection onto the $xy$-plane must not be parallel or antiparallel to the initial Bloch vector [$\sin(\phi - \varphi) \neq 0$]. Under these constraints, solving Eq.~\eqref{eq:critical_angle} yields the critical rotation angle:
\begin{equation}
    \chi_\mathrm{c} = \begin{cases}
        \pi &\text{if $\cos\vartheta = 0$ or $\cos(\phi - \varphi) = 0$}\\
        2\arctan\left[-\dfrac{\tan(\phi - \varphi)}{\cos\vartheta}\right] &\text{else}
    \end{cases}.
\label{eq:critical_angle_general}
\end{equation}

\subsection{Critical rotation angle by imposing a symmetrical distribution}\label{chp:critical_angle_konno}  % {{{*
As proved by Konno~\cite{konno_2005_limit}, the position expectation value of a 1D DTQW vanishes at all times $t \geq 0$ if and only if its probability distribution remains symmetric at all times $t \geq 0$. Relying on this equivalence, we now derive the critical rotation angle $\chi_\mathrm{c}$ by invoking the following theorem.
\begin{theorem}\label{th:SymmetricProbDistrib}
A 1D discrete-time quantum walk described by a two-dimensional coin operator $C$ and an initial coin state $\ket{c_0} = \cos(\theta/2) \ket{0} + e^{i\phi}\sin(\theta/2) \ket{1}$ has a symmetric probability distribution at all times $t$ [i.e., $P(x, t) = P(-x, t)$ for all $x$ and $t>0$] if and only if:
\begin{enumerate}
    \item $\theta = \pi/2$, and
    \item $\Re\!\Big[C_{00}^*C_{01}e^{i\phi}\Big] = 0$.
\end{enumerate}
\end{theorem}
Theorem~\ref{th:SymmetricProbDistrib} is a slightly modified and generalized version of a theorem first provided and proved by Konno~\cite{konno_2005_limit}, and more clearly presented by Hantzko et al. in Ref.~\cite{hantzko_2026_classification}. The original statement assumes $C_{00}$ and $C_{01}$ to be non-zero. Therefore, to extend it to arbitrary coin operators, it is enough to consider the two excluded cases: diagonal and anti-diagonal coin operators.
\begin{proof}
When $C_{00}$ and $C_{01}$ are non-zero, the result follows directly from Ref.~\cite{konno_2005_limit}. If the coin operator is diagonal, $C_{01}=0$, the second condition is trivially satisfied. The evolution produces a ballistic split,
$\ket{\psi_t}=\cos(\theta/2)C_{00}^t\ket{t,0}
+e^{i\phi}\sin(\theta/2)(C_{00}^*)^t\ket{-t,1}$,
so the probabilities at $x=t$ and $x=-t$ are $|\cos(\theta/2)|^2$ and $|\sin(\theta/2)|^2$, respectively. Hence, the distribution is symmetric for all $t$ if and only if $\theta=\pi/2$.

If the coin operator is anti-diagonal, $C_{00}=0$, the second condition is again trivially satisfied. At even times the walker is localized at the origin, so the distribution is symmetric for any initial coin state. At odd times, the walker is split between $x=1$ and $x=-1$ with probabilities $|\sin(\theta/2)|^2$ and $|\cos(\theta/2)|^2$, respectively. Thus, the distribution is symmetric for all $t$ if and only if $\theta=\pi/2$. This proves that the theorem also holds for diagonal and anti-diagonal coin operators.
\end{proof}

We now apply Theorem~\ref{th:SymmetricProbDistrib} to the coin operator $C=\exp[-i(\chi/2)\hat n\cdot\vec\sigma]$, with fixed rotation axis $\hat n(\vartheta,\varphi)$. Its relevant matrix elements are $C_{00}=\cos(\chi/2)-i\cos\vartheta\sin(\chi/2)$ and $C_{01}=-ie^{-i\varphi}\sin\vartheta\sin(\chi/2)$. Condition (1) requires the initial coin state to lie on the equator of the Bloch sphere, $\theta=\pi/2$, while condition (2) gives
\[
    \sin\vartheta\sin(\chi/2)
    \left[
        \sin(\phi-\varphi)\cos(\chi/2)
        +
        \cos(\phi-\varphi)\cos\vartheta\sin(\chi/2)
    \right]
    =
    0,
\]
which is exactly Eq.~\eqref{eq:critical_angle} derived in Sec.~\ref{chp:critical_angle_transport_vector}. Therefore, for a fixed configuration $\{\vartheta,\varphi,\theta=\pi/2,\phi\}$, the critical rotation angle $\chi_\mathrm{c}$ is the unique non-trivial solution to Eq.~\eqref{eq:critical_angle}. Hence, the  critical angle obtained from imposing a symmetric probability distribution, and therefore a vanishing expectation value of the position, for all times $t\geq 0$ is equivalent to the expression obtained by setting the asymptotic velocity to zero for initial coin states on the $xy$-plane.

% *}}}
% *}}}
\section{Detailed calculation of $\mathbb{P}(L\circ L = W) + \mathbb{P}(W\circ W = L)$} % {{{*
\label{sec:joint_prob_derivation}
In this section, we derive 
the probability for Parrondo's paradox to emerge, that is, 
$\mathbb{P}(L\circ L = W) + \mathbb{P}(W\circ W = L)$,
in the case of coin composition. We assume the constituent coin operators 
and the initial state of the coin are drawn uniformly with respect to the 
Haar measure. We first compute the losing-to-winning contribution,
$\mathbb{P}(L\circ L = W) \equiv \mathbb{P}\qty[\vmcT(SC_1) \cdot \vec{r}_0 < 0, \; \vmcT(SC_2) \cdot \vec{r}_0 < 0, \; \vmcT(SC_2C_1) \cdot \vec{r}_0 > 0]$. Here, $\vmcT(SC_1)$ and $\vmcT(SC_2)$ are the transport vectors associated with independent Haar-distributed SU(2) coin operators $C_1$ and $C_2$, $\vmcT(SC_2C_1)$ corresponds to the composite coin operator, and $\vec{r}_0$ is the Bloch vector of an initial state of the coin, uniformly sampled over the Bloch sphere.

We first evaluate the marginal probability of a single losing strategy, $\mathbb{P}(\vmcT(SC) \cdot \vec{r}_0 < 0)$.
Let \(C\in \mathrm{SU}(2)\), and denote its first row by the complex pair
$(C_{00},C_{01})$. Unitarity imposes
$|C_{00}|^2+|C_{01}|^2=1$, while the condition $\det C=1$
fixes the second row as $(-C_{01}^*,C_{00}^*)$. Thus each $C\in\mathrm{SU}(2)$ is uniquely represented by a normalized pair $(C_{00},C_{01})\in\mathbb{C}^2$, i.e., by a point on the unit sphere in four dimensions. The standard Hopf map projects the uniform distribution of $(C_{00}, C_{01})$ from this four-dimensional sphere to a uniform distribution of coordinates $(x,y,z)$ on a standard unit sphere in three dimensions, where $x = 2\Re(C_{01}C_{00}^*)$, $y = 2\Im(C_{01}C_{00}^*)$, and $z = |C_{00}|^2 - |C_{01}|^2$~\cite{bengtsson_2006_geometry}. Substituting $|C_{00}|^2 = (1+z)/2$ and $|C_{01}| = \sqrt{(1-z)/2}$ into the transport vector expression in \Eref{eq:transport_vector_integrated}, we obtain:
\begin{equation}
  \vmcT(SC) = \frac{\sqrt{2}}{2\qty(\sqrt{2}+\sqrt{1-z})} \mqty( x \\ -y \\ 1 + z ) = \frac{\sqrt{2}}{2\qty(\sqrt{2}+\sqrt{1-z})} (\vec{P} - \vec{S}),
\end{equation}
where $\vec{P} = (x, -y, z)^T$ and $\vec{S} = (0,0,-1)^T$ is the south pole.
If $C$ is randomly sampled from the Haar measure, the Hopf map makes $\vec{P}$
uniformly distributed over the unit sphere.
Since the scalar prefactor is strictly positive for all $z \in [-1,1]$, the sign of the projection onto $\vec{r}_0$ depends entirely on the vector difference: $\vmcT(SC) \cdot \vec{r}_0 < 0 \iff (\vec{P} - \vec{S}) \cdot \vec{r}_0 < 0$. Using the linearity of the dot product and evaluating $\vec{S} \cdot \vec{r}_0 = -\cos\theta$, where $\theta$ is the polar angle of $\vec{r}_0$, the condition simplifies to $\vec{P} \cdot \vec{r}_0 < -\cos\theta$. Geometrically, this restricts $\vec{P}$ to a spherical cap, the portion of the unit sphere cut off by the plane $\vec{P}\cdot\vec{r}_0=-\cos\theta$, of height $1-\cos\theta$. Because $\vec{P}$ is uniformly distributed over the sphere, the conditional probability for a fixed $\vec{r}_0$ is the ratio between the area of this cap and the total area of the unit sphere $p(\theta) = (1-\cos\theta)/2 \equiv u$. Integrating over all possible initial states $\vec{r}_0$ uniformly distributed on the Bloch sphere, whose polar angle has a probability density $f(\theta) = \frac{1}{2}\sin\theta$, we obtain the unconditional marginal probability, that is, the overall probability of an individual strategy yielding a losing outcome:
\begin{equation}
  \mathbb{P}(L) = \int_{0}^{\pi} p(\theta) f(\theta) \dd{\theta} = \int_{0}^{1} u \dd{u} = \frac{1}{2}.
\end{equation}

Next, we consider the probability that two individual strategies are losing, $\mathbb{P}\qty[\vmcT(SC_1) \cdot \vec{r}_0 < 0, \; \vmcT(SC_2) \cdot \vec{r}_0 < 0]$. The independence of the two Haar-distributed coin operators implies independence only after conditioning on a fixed initial state $\vec{r}_0$. For such a fixed $\vec{r}_0$, the two losing events are independent and have the same conditional probability $p(\theta)=u$. Hence,
$\mathbb{P}\qty[\vmcT(SC_1) \cdot \vec{r}_0 < 0, \;\vmcT(SC_2) \cdot \vec{r}_0 < 0 \mid \vec{r}_0] = p(\theta)^2 = u^2$.  The unconditional probability is then obtained by averaging this conditional joint probability over the uniformly distributed initial state. Therefore,
\begin{equation}\label{eq:p_LL}
  \mathbb{P}(L,L) = \int_{0}^{1} u^2 \dd{u} = \frac{1}{3}.
\end{equation}
This result differs from $\mathbb{P}(L)^2=1/4$ because the two losing events share the same random initial state $\vec{r}_0$; they are conditionally independent for fixed $\vec{r}_0$, but not marginally independent after the average over $\vec{r}_0$.

Finally, we use the pairwise probabilities derived above to determine
the probability of the $L\circ L=W$ contribution.
The three relevant outcomes are the signs of
$\vmcT(SC_1)\cdot\vec{r}_0$, $\vmcT(SC_2)\cdot\vec{r}_0$, and
$\vmcT(SC_2C_1)\cdot\vec{r}_0$. For fixed $\vec{r}_0$, any pair among
$C_1$, $C_2$, and $C_2C_1$ is distributed as a pair of independent Haar-distributed
$\mathrm{SU}(2)$ matrices.
The three sign events, however, are not mutually independent, because
$C_2C_1$ is determined by $C_1$ and $C_2$.

The remaining step is to reconstruct the full joint probability from these pairwise probabilities. Since $\vec{r}_0$ is sampled uniformly on the Bloch sphere, the change of variables $\vec{r}_0\mapsto-\vec{r}_0$ leaves the measure unchanged and interchanges $L\leftrightarrow W$ in all three outcomes. Hence, each configuration has the same probability as its complementary configuration:
\begin{equation}
\begin{aligned}
    \mathbb{P}(L\circ L = L)&=\mathbb{P}(W\circ W = W),\\
    \mathbb{P}(L\circ L = W)&=\mathbb{P}(W\circ W = L),\\
    \mathbb{P}(L\circ W = L)&=\mathbb{P}(W\circ L = W),\\
    \mathbb{P}(W\circ L = L)&=\mathbb{P}(L\circ W = W).
\end{aligned}
\label{eq:complementary_configurations}
\end{equation}

Then, we use the fact that $\mathbb{P}(L,L)$ in \Eref{eq:p_LL} applies to
any pair of strategies whose coin operators are independently sampled and
both yield losing outcomes.
Therefore, if both individual strategies, the first individual and combined strategy, or the second individual and combined strategy, are losing, we can write, respectively:
\begin{equation}
\begin{aligned}
    \mathbb{P}(L, L) = \mathbb{P}(L\circ L = L)+\mathbb{P}(L\circ L = W) = \frac{1}{3}, \\
    \mathbb{P}(L, L) = \mathbb{P}(L\circ L = L)+\mathbb{P}(L\circ W = L) = \frac{1}{3}, \\
    \mathbb{P}(L, L) = \mathbb{P}(L\circ L = L)+\mathbb{P}(W\circ L = L) = \frac{1}{3}.
\end{aligned}
\label{eq:pairwise_losing_marginals}
\end{equation}
In turn, these imply:
\begin{equation}
    \mathbb{P}(L\circ L = W)
    =
    \mathbb{P}(L\circ W = L)
    =
    \mathbb{P}(W\circ L = L).
\label{eq:two_losing_configurations_equal}
\end{equation}
By Eqs.~\eqref{eq:complementary_configurations} and
\eqref{eq:two_losing_configurations_equal}, the three configurations with
exactly one losing outcome have this same probability. Therefore,
normalization of the eight possible outcomes gives
\begin{equation}
    2\,\mathbb{P}(L\circ L = L)
    +
    6\,\mathbb{P}(L\circ L = W)
    =
    1.
\label{eq:joint_probability_normalization}
\end{equation}
Together with the first relation in Eq.~\eqref{eq:pairwise_losing_marginals},
\begin{equation}
    \mathbb{P}(L\circ L = L)
    +
    \mathbb{P}(L\circ L = W)
    =
    \frac{1}{3},
\label{eq:first_pairwise_losing_marginal}
\end{equation}
Eq.~\eqref{eq:joint_probability_normalization} yields
$\mathbb{P}(L\circ L = L)=1/4$. Hence,
the probability of the $L\circ L=W$ contribution is
\begin{equation}\label{eq:parrondo_joint_probability}
    \mathbb{P}(L\circ L = W) = \frac{1}{12}.
\end{equation}
By Eq.~\eqref{eq:complementary_configurations},
$\mathbb P(W\circ W=L)=\mathbb P(L\circ L=W)=1/12$. Therefore,
\begin{equation}
\mathbb P(L\circ L=W)+\mathbb P(W\circ W=L)=\frac{1}{6}.
\end{equation}

% *}}}
\section{Derivation of expressions for~\Fref{fig:PhaseDiagrams}} % {{{*
\label{SM:expr_fig3}
This section specializes the general transport vector and critical-angle results of Secs.~\ref{sm:transport_vector_integral} and~\ref{SM:critical_rot_angle} to the coin operator and initial state used in \Fref{fig:PhaseDiagrams}, giving the closed-form expressions plotted there and the areas of the
outcome-map regions used to quote the probabilities in the main text.
\subsection{Configuration} % {{{*
\Fref{fig:PhaseDiagrams} uses the Hadamard rotation axis, $\hat n(\vartheta=\pi/4,\varphi=0)=(1,0,1)/\sqrt2$, and an initial coin state on the equator with $\vec r_0(\theta=\pi/2,\phi=\pi/4)=(1,1,0)/\sqrt2$, i.e., $\ket{c_0}=(\ket0+e^{i\pi/4}\ket1)/\sqrt2$.
% *}}}
\subsection{Transport vector and asymptotic velocity} % {{{*
Substituting $\vartheta=\pi/4$ and $\varphi=0$ into the general transport vector, \Eref{eq:transp_vec}, gives
\begin{equation}\label{eq:transp_vec_Hadamard}
    \vmcT(\chi) = \frac{1}{2+\sqrt{2} \sin(\chi/2)}
    \begin{pmatrix}
        \sin^2(\chi/2) \\[0.25em]
        \sqrt{2} \sin(\chi/2) \cos(\chi/2) \\[0.25em]
        2 - \sin^2(\chi/2)
    \end{pmatrix}.
\end{equation}
Likewise, substituting $\vartheta=\pi/4$, $\varphi=0$, and $\phi=\pi/4$ into the general asymptotic velocity, \Eref{eq:asymp_vel_xyplane}, gives
\begin{equation}
    \overline{v}(\chi) = \frac{2\sin(\chi/2)\cos(\chi/2) + \sqrt{2}\sin^2(\chi/2)}{4 + 2\sqrt{2}\sin(\chi/2)},
\end{equation}
which, after the double-angle substitutions $2\sin(\chi/2)\cos(\chi/2)=\sin\chi$ and $2\sin^2(\chi/2)=1-\cos\chi$ followed by the harmonic-addition identity $a\sin\chi+b\cos\chi=\sqrt{a^2+b^2}\sin[\chi+\arctan(b/a)]$ with $a=1$, $b=-\sqrt2/2$, simplifies to
\begin{equation}\label{eq:asymp_vel_Hadamard}
    \overline{v}(\chi) = \frac{\sqrt{3}\sin\left[\chi - \arctan\left(\frac{1}{\sqrt{2}}\right)\right] + 1}{4\left[\sin(\chi/2) + \sqrt{2}\,\right]}.
\end{equation}
This represents the dashed black line in the main panel of \Fref{fig:PhaseDiagrams}.
% *}}}
\subsection{Critical rotation angle} % {{{*
Substituting the same $\vartheta=\pi/4$, $\varphi=0$, $\phi=\pi/4$ into the general critical angle, \Eref{eq:critical_angle_general}, gives the principal value $2\arctan(-\sqrt2)\in(-\pi,0)$. Shifting by $2\pi$ to the convention $\chi_\mathrm{c}\in(0,2\pi)$ used in the main text yields
\begin{equation}\label{eq:critical_angle_Hadamard}
    \chi_\mathrm{c} = 2\left[\pi - \arctan(\sqrt{2})\right].
\end{equation}
% *}}}
\subsection{Outcome-map areas}
Table~\ref{tab:areas_side_by_side} gives the areas of the six different regions of the $(\chi_1,\chi_2)$ outcome map introduced in the main text. Since coin operators with the same rotation axis commute, $L \circ W = L$ also includes $W \circ L = L$ (and likewise for $L \circ W = W$). The areas are given as a function of $\chi_\mathrm{c}$ distinguishing which side of $\pi$ it falls on. Also, the two tables correspond to the two possible ranges of
$(\phi-\varphi)\bmod2\pi$, which determine whether
$0<\chi<\chi_\mathrm{c}$ is winning and
$\chi_\mathrm{c}<\chi<2\pi$ is losing, or vice versa, for both the
individual angles $\chi_{1,2}$ and the combined angle
$(\chi_1+\chi_2)\bmod2\pi$. Summing any pair of paradoxical or intuitive entries reproduces the compact formulas quoted in the main text, e.g.\ $\mathbb P(L\circ L=W)+\mathbb P(W\circ W=L)\propto\chi_\mathrm{c}(2\pi-\chi_\mathrm{c})$.

For the configuration of \Fref{fig:PhaseDiagrams}, $(\phi-\varphi)\bmod2\pi=\pi/4\in(0,\pi)$ and, from \Eref{eq:critical_angle_Hadamard}, $\chi_\mathrm{c}=2[\pi-\arctan(\sqrt2)]\in(\pi,2\pi)$, so the areas shown in its inset are those of the $\chi_\mathrm{c}\in(\pi,2\pi)$ column of \Tref{tab:areas_side_by_side}(b).

\begin{table}[htpb]
\centering
\renewcommand{\arraystretch}{1.5}

\caption{Outcome map areas for the two branches of $(\phi - \varphi) \!\! \mod{2\pi}$.}
\label{tab:areas_side_by_side}

\begin{minipage}[t]{0.49\textwidth}
\centering
\vspace{0pt}
\textbf{(a)} $(\phi - \varphi) \!\! \mod{2\pi} \in (0, \pi)$

\vspace{0.5em}

\resizebox{\linewidth}{!}{%
\begin{tabular}{lcc}
    \hline
    Region & $\chi_\mathrm{c} \in (0, \pi]$ & $\chi_\mathrm{c} \in (\pi, 2\pi)$ \\
    \hline
    $L \circ L = L$ & $\frac{1}{2}(2\pi - \chi_\mathrm{c})^2 + 2\qty(\chi_\mathrm{c} - \pi)^2$ & $\frac{1}{2}(2\pi - \chi_\mathrm{c})^2$ \\
    $L \circ L = W$ & $\frac{1}{2}(2\pi - \chi_\mathrm{c})^2 - 2\qty(\chi_\mathrm{c} - \pi)^2$ & $\frac{1}{2}(2\pi - \chi_\mathrm{c})^2$ \\
    $W \circ W = W$ & $\frac{1}{2}\chi_\mathrm{c}^2$ & $\frac{1}{2}\chi_\mathrm{c}^2 + 2\qty(\chi_\mathrm{c} - \pi)^2$ \\
    $W \circ W = L$ & $\frac{1}{2}\chi_\mathrm{c}^2$ & $\frac{1}{2}\chi_\mathrm{c}^2 - 2\qty(\chi_\mathrm{c} - \pi)^2$ \\
    $L \circ W = L$ & $2\qty[\frac{1}{2}(2\pi - \chi_\mathrm{c})^2 - 2\qty(\chi_\mathrm{c} - \pi)^2]$ & $(2\pi - \chi_\mathrm{c})^2$ \\
    $L \circ W = W$ & $\chi_\mathrm{c}^2$ & $2\qty[\frac{1}{2}\chi_\mathrm{c}^2 - 2\qty(\chi_\mathrm{c} - \pi)^2]$ \\
    \hline
\end{tabular}%
}
\end{minipage}
\hfill
\begin{minipage}[t]{0.49\textwidth}
\centering
\vspace{0pt}
\textbf{(b)} $(\phi - \varphi) \!\! \mod{2\pi} \in (\pi, 2\pi)$

\vspace{0.5em}

\resizebox{\linewidth}{!}{%
\begin{tabular}{lcc}
  \hline
  Region & $\chi_\mathrm{c} \in (0, \pi]$ & $\chi_\mathrm{c} \in (\pi, 2\pi)$ \\
  \hline
  $L \circ L = L$ & $\frac{1}{2}\chi_\mathrm{c}^2$ & $\frac{1}{2}\chi_\mathrm{c}^2 + 2\qty(\chi_\mathrm{c} - \pi)^2$ \\
  $L \circ L = W$ & $\frac{1}{2}\chi_\mathrm{c}^2$ & $\frac{1}{2}\chi_\mathrm{c}^2 - 2\qty(\chi_\mathrm{c} - \pi)^2$ \\
  $W \circ W = W$ & $\frac{1}{2}(2\pi - \chi_\mathrm{c})^2 + 2\qty(\chi_\mathrm{c} - \pi)^2$ & $\frac{1}{2}(2\pi - \chi_\mathrm{c})^2$ \\
  $W \circ W = L$ & $\frac{1}{2}(2\pi - \chi_\mathrm{c})^2 - 2\qty(\chi_\mathrm{c} - \pi)^2$ & $\frac{1}{2}(2\pi - \chi_\mathrm{c})^2$ \\
  $L \circ W = L$ & $\chi_\mathrm{c}^2$ & $2\qty[\frac{1}{2}\chi_\mathrm{c}^2 - 2\qty(\chi_\mathrm{c} - \pi)^2]$ \\
  $L \circ W = W$ & $2\qty[\frac{1}{2}(2\pi - \chi_\mathrm{c})^2 - 2\qty(\chi_\mathrm{c} - \pi)^2]$ & $(2\pi - \chi_\mathrm{c})^2$ \\
  \hline
\end{tabular}%
}
\end{minipage}
\end{table}

\end{document}